\documentclass[
aip,
cha,
amsmath,amssymb,
reprint,
floatfix
]{revtex4-1}

\usepackage{graphicx}
\usepackage{dcolumn}
\usepackage{bm}
\usepackage[utf8]{inputenc}
\usepackage[T1]{fontenc}
\usepackage{mathptmx}
\usepackage{etoolbox}
\usepackage{amsfonts}
\usepackage{amsthm}
\usepackage{mathrsfs}
\usepackage{xcolor}
\usepackage{textcomp}
\usepackage{booktabs}
\usepackage{enumitem}
\usepackage{mathtools}
\usepackage{array}
\usepackage{placeins}

\makeatletter
\def\@email#1#2{%
	\endgroup
	\patchcmd{\titleblock@produce}
	{\frontmatter@RRAPformat}
	{\frontmatter@RRAPformat{\produce@RRAP{*#1\href{mailto:#2}{#2}}}\frontmatter@RRAPformat}
	{}{}
}%
\makeatother

\newtheorem{theorem}{Theorem}[section]
\newtheorem{lemma}[theorem]{Lemma}

\newtheorem{definition}{Definition}[section]

\begin{document}
	
	% Two-column mathematical layout tuning: keep normal font size, but give
	% multi-line displays enough vertical room and permit sensible page breaks.
	\allowdisplaybreaks[3]
	\setlength{\jot}{4pt}
	\setlength{\abovedisplayskip}{7pt plus 2pt minus 2pt}
	\setlength{\belowdisplayskip}{7pt plus 2pt minus 2pt}
	\setlength{\abovedisplayshortskip}{4pt plus 2pt minus 1pt}
	\setlength{\belowdisplayshortskip}{5pt plus 2pt minus 1pt}
	
	\title{Topology-Biased Resource Constraints Shape Synchronization Pathways in Hindmarsh–Rose Oscillator Networks}
	
	\author{Zhouqi Li}
	\affiliation{School of Computer and Information Management, Inner Mongolia University of Finance and Economics, Hohhot, China}
	
	\author{Xiaoyan He}
	\author{Yuanhong Bi}
	\affiliation{School of Statistics and Mathematics, Inner Mongolia University of Finance and Economics, Hohhot, China}
	
	\author{Zengping Zhang*}
	\email[Corresponding Zengping Zhang: ]{zzp@imufe.edu.cn}
	\affiliation{School of Computer and Information Management, Inner Mongolia University of Finance and Economics, Hohhot, China}
	
	\date{\today}
	
	\begin{abstract}
		In oscillator networks sustained by finite resources, synchronization can depend on both the total resource and its spatial distribution. We study a duplex system whose activity layer consists of chaotic Hindmarsh--Rose oscillators and whose transport layer redistributes a conserved resource through a degree-biased Markov process. The stationary resource field is characterized analytically, and its existence, uniqueness, and convergence are established. By embedding this field into a local adaptive feedback law, the available resource is converted into node-dependent dissipation, for which Lyapunov analysis guarantees convergence to the synchronization manifold. Numerical results show that topology bias reorganizes the transient route to synchronization. Weak bias produces an almost collective contraction, whereas intermediate bias creates a hub-initiated recruitment hierarchy that extends toward middle-degree and peripheral nodes. Under stronger bias, the degree hierarchy becomes more pronounced while peripheral recruitment slows because adaptive dissipation is concentrated on structurally privileged nodes. Across the explored parameter range, this localization--coverage tradeoff is accompanied by a non-monotonic synchronization response at fixed total resource. The largest Lyapunov exponent remains positive after synchronization and the correlation dimension changes only modestly, consistent with suppression of transverse deviations while chaotic motion is retained on the synchronization manifold.
	\end{abstract}
	
	
	\maketitle
	
	%\linenumbers

	\begin{quotation}
		Oscillator networks with limited resources can synchronize differently depending on where those resources are available. Here, a conserved resource is redistributed according to network structure and weights local adaptive dissipation in a chaotic oscillator network. Moderate localization leads to degree-ordered recruitment initiated by highly connected nodes, whereas stronger localization further separates the degree classes but delays peripheral recruitment. The results therefore distinguish the strength of a synchronization hierarchy from its effectiveness in recruiting the network as a whole.
	\end{quotation}

	\section{Introduction}
	
	Synchronization in networks of nonlinear oscillators is a fundamental problem in complex-system dynamics. Coupled oscillatory units are widely found in chemical reactions \cite{kiss2002emerging}, laser systems \cite{nair2021disorder}, Josephson junctions \cite{wiesenfeld1996synchronization}, power grids \cite{motter2013spontaneous}, and biological rhythms \cite{winfree1967biological}. For example, synchronization plays an important role in brain activity, including motor cortical coordination and memory-related processes \cite{riehle1997spike,fell2011role}, and is closely associated with neurological diseases such as Parkinson's disease, Alzheimer's disease, and epilepsy \cite{touboul2020noise,pusil2019aberrant,zhang2020spontaneous}.
	
	The dynamics underlying synchronization are influenced by various factors, including time delays, noise, coupling forms, and network topology \cite{lakshmanan2016dynamical,lepek2018spatial,sun2011effects,sun2010effects,breunung2025noise,jalili2009synchronizing,parastesh2022synchronization,gerster2020fitzhugh,huang2025effects,afifurrahman2020stability}. For example, early studies using the Hindmarsh--Rose model showed that increasing coupling can simplify firing-interval dynamics, enlarge phase-locked neuronal populations, and lead to synchronized chaotic states \cite{wang1993dynamical}. Real oscillator networks require continuous resource exchange to sustain ordered activity far from equilibrium. Maintaining synchronization generally requires a nonzero mean resource throughput \cite{sarasola2004energy,torrealdea2009energy,moujahid2011efficient}. Such resource dependence is manifested in systems across different fields. For instance, fMRI studies have shown that variations in blood oxygenation are closely associated with the level of neuronal activity \cite{logothetis2001neurophysiological,heeger2002does,sheth2004linear,allen2007transcranial,logothetis2008we}. In real systems, however, the available resources are not inexhaustible. Understanding how finite resource constraints shape synchronization in oscillator networks therefore constitutes an important dynamical problem.
	
	Resource constraints have recently been incorporated into oscillator networks in several forms, including coupling mediated by shared resources or dynamical environments \cite{li2012quorum,singla2016environmental,postnov2005oscillator,resmi2010environment}, constraints arising from resource consumption, activity-dependent resource pools, or excitability limitations \cite{kromawiley2021resource,franovic2022resources,frolov2021extreme}, and resource costs or finite coupling budgets \cite{zhang2020energy,liang2022time,zhang2021designing}. Resource redistribution has also been modeled as an independent dynamical process on a separate network layer and coupled to the oscillator dynamics \cite{virkar2016resource,nicosia2017collective}.
	
	Among these studies, Nicosia \textit{et al.} \cite{nicosia2017collective} placed Kuramoto phase oscillators and random-walk transport on two layers of a multiplex network and reported bistability, hysteresis, and explosive synchronization. In their model, resources mainly modulate the natural frequencies of the oscillators, while the relationship between resource allocation and synchronization stability remains to be further explored. In a different physical perspective, resource-dependent feedback can provide additional dissipation that suppresses dynamical divergence and stabilizes synchrony \cite{yamakou2020chaotic,kobiolka2025reduced}. This raises a separate question: when the total resource is fixed, does stronger topology bias only change the synchronization time, or can it also alter the temporal ordering and extent of recruitment across the network?
	
	We consider a setting in which a fixed amount of resource is competitively distributed among the oscillators, producing a spatially heterogeneous allocation. For chaotic oscillators, transverse perturbations away from the synchronization manifold may be amplified by the intrinsic dynamics. Here the local resource share weights an adaptive feedback term applied to each oscillator, providing additional dissipation that suppresses transverse deviations and promotes synchronization. The model therefore links topology-dependent resource transport to synchronization stability through resource-dependent dissipation.
	
	In this paper, we consider a resource-constrained network of chaotic Hindmarsh--Rose oscillators coupled to a topology-dependent transport layer. A degree-biased Markov process redistributes a finite shared resource, and its stationary distribution determines the local availability of adaptive dissipation. We establish the existence, uniqueness, and convergence of this stationary resource field and prove synchronization of the resulting closed-loop network by Lyapunov analysis. We then examine how redistributing the same total resource changes the transient organization of synchronization. Weak bias is associated with nearly collective contraction, while intermediate bias produces a hub-initiated ordering of recruitment. Under stronger bias, the degree hierarchy becomes more pronounced but recruitment of peripheral nodes slows. Thus, hierarchy strength and network-wide recruitment effectiveness do not increase together over the explored bias range. The synchronized collective state retains the intrinsic chaotic dynamics of the individual oscillator.
	
	The remainder of this paper is organized as follows. Section~2 introduces the model. Section~3 proves the uniqueness and convergence of the resource field under degree bias and establishes synchronization under the proposed adaptive feedback through Lyapunov analysis. Section~4 presents the numerical results and analyzes the dynamics induced by different resource-allocation regimes. Section~5 concludes the paper.

	\section{Modeling and Preparatory Knowledge}
	\subsection{Algebraic Graph Theory}
	
	Consider a network \(G=(V,E)\) consisting of \(N\) neurons, where \(V=\{1,2,\dots,N\}\) is the node set and \(E\subseteq V\times V\) denotes the set of synaptic connections. The network topology is described by the adjacency matrix \(A=[a_{ij}]_{N\times N}\), where \(a_{ij}=1\) if nodes \(i\) and \(j\) are connected and \(a_{ij}=0\) otherwise. The degree of node \(i\) is \(k_i=\sum_{j=1}^{N}a_{ij}\), and the corresponding degree matrix is \(D=\mathrm{diag}(k_1,k_2,\dots,k_N)\). The Laplacian matrix is defined as \(L=D-A=[l_{ij}]\in\mathbb{R}^{N\times N}\), where \(l_{ii}=k_i\) and \(l_{ij}=-a_{ij}\) for \(i\neq j\).
	
	For an undirected network, the Laplacian matrix has two fundamental properties relevant to synchronization analysis:
	
	\begin{enumerate}[label=(\arabic*)]
		
		\item \textbf{Positive Semi-Definiteness.}
		For any vector \(\varphi=[\varphi_1,\varphi_2,\dots,\varphi_N]^T\in\mathbb{R}^N\),
		\[
		\varphi^T L\varphi
		=
		\frac{1}{2}\sum_{i=1}^{N}\sum_{j=1}^{N}
		a_{ij}(\varphi_i-\varphi_j)^2
		\geq 0.
		\]
		This property reflects the dissipative nature of diffusive coupling, which tends to reduce state differences between interconnected nodes.
		
		\item \textbf{Zero Eigenvalue and Connectivity.}
		The Laplacian matrix satisfies \(L\mathbf{1}=0\), where \(\mathbf{1}=[1,1,\dots,1]^T\). The multiplicity of the zero eigenvalue equals the number of connected components of the graph. In particular, for a connected network, the zero eigenvalue is simple. Moreover, \(L\mathbf{1}=0\) implies that the diffusive coupling vanishes on the synchronization manifold
		\[
		X_1(t)=X_2(t)=\cdots=X_N(t)=s(t).
		\]
		
	\end{enumerate}
	
	\subsection{Model Description}
	
	Consider a network of \(N\) Hindmarsh--Rose (HR) neurons interconnected through a connected undirected graph \(G=(V,E)\). The state of the \(i\)-th neuron is denoted by
	\[
	X_i(t)=\bigl[x_i(t),y_i(t),z_i(t)\bigr]^\top\in\mathbb{R}^3,
	\qquad i=1,2,\dots,N.
	\]
	The neurons are coupled through the membrane-potential channel. Accordingly, the network dynamics are described by
	\begin{equation}\label{eq:net_general}
		\dot{X}_i(t)
		=
		F\bigl(X_i(t)\bigr)
		-\sigma\sum_{j=1}^{N}l_{ij}\Gamma X_j(t)
		+B u_i(t),
	\end{equation}
	where \(\sigma>0\) denotes the coupling strength, \(L=[l_{ij}]\) is the Laplacian matrix of the underlying network, and
	\(\Gamma=\mathrm{diag}(1,0,0)\) specifies the coupling channel. The scalar control input \(u_i(t)\in\mathbb{R}\) acts on the same channel through \(B=[1,0,0]^\top\).
	
	The Hindmarsh--Rose model \cite{hindmarsh1984model} is adopted to describe the intrinsic neuronal dynamics. Compared with the detailed Hodgkin--Huxley model \cite{hodgkin1952quantitative}, the HR model retains the essential fast--slow dynamics of neuronal firing with considerably lower computational complexity, while reproducing a variety of dynamical behaviors including periodic spiking, bursting, and chaotic oscillations \cite{izhikevich2003model}. The intrinsic vector field is given by
	\begin{equation}\label{eq:HR_F}
		F(X_i)=
		\begin{bmatrix}
			y_i-a x_i^{3}+b x_i^{2}-z_i+I_{\mathrm{ext}}\\
			c-d_0 x_i^{2}-y_i\\
			r\bigl[s_0(x_i+x_0)-z_i\bigr]
		\end{bmatrix}.
	\end{equation}
	Here, \(x_i\) represents the membrane potential, while \(y_i\) and \(z_i\) denote the fast and slow recovery variables, respectively. The parameters \(a\), \(b\), \(c\), and \(d_0\) determine the fast neuronal dynamics; \(r\) controls the time scale of the slow variable; \(s_0\) and \(x_0\) regulate the slow recovery process; and \(I_{\mathrm{ext}}\) denotes the external input current.
	
	To characterize the collective behavior of the network, define the mean state
	\begin{equation}\label{eq:Xbar_def}
		\bar X(t)=\frac{1}{N}\sum_{i=1}^{N}X_i(t),
	\end{equation}
	and the synchronization error of neuron \(i\) as
	\begin{equation}\label{eq:ei_def}
		e_i(t)=X_i(t)-\bar X(t).
	\end{equation}
	
	It follows immediately that 
	\[
	\sum_{i=1}^{N}e_i(t)=0.
	\]
	
	\begin{definition}
		The network \eqref{eq:net_general} is said to achieve complete 
		synchronization if
		\[
		\lim_{t\to\infty}\|e_i(t)\|=0,
		\qquad i=1,2,\dots,N.
		\]
	\end{definition}

	\subsection{Topology-Driven Resource Allocation}
	
	Synchronization generally requires additional control effort, yet finite resources cannot be distributed uniformly across all neurons. We therefore consider a topology-dependent allocation mechanism in which neurons compete for a shared resource, and their access to this resource is modulated by their structural positions in the network.
	
	Structural centrality is known to influence collective neuronal dynamics. Neuronal hub cells can contribute to the initiation and propagation of epileptic activity\cite{morgan2008nonrandom,cossart2014operational}. Motivated by these observations, node degree is adopted here to characterize the competitive preference for the available resource.

	\begin{figure}[t]
		\centering
		\includegraphics[width=0.5\textwidth]{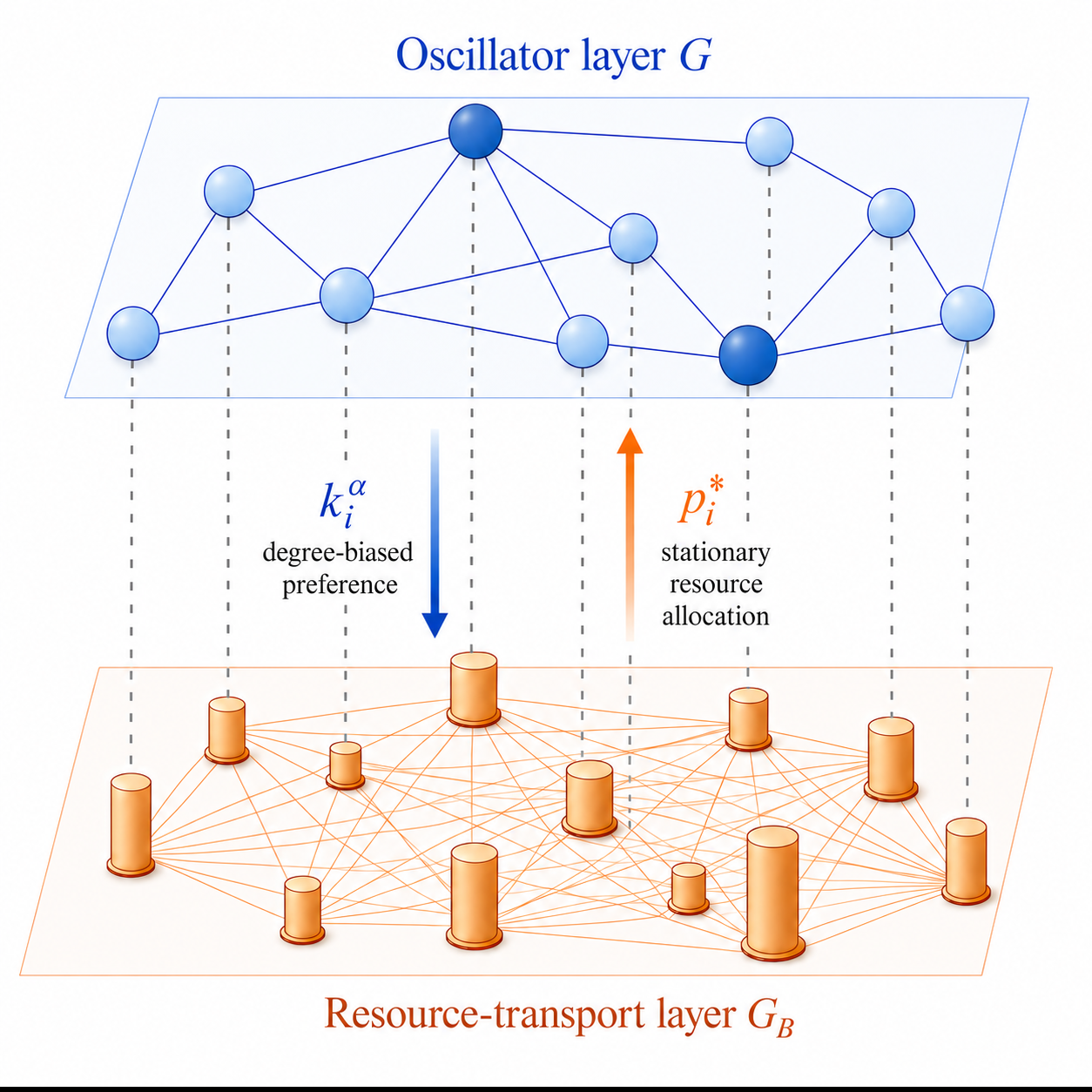}
		\caption{
			Schematic of the two-layer network. The upper layer describes electrical interactions among HR neurons, while the lower layer represents the redistribution of a finite resource. The transport topology determines the available resource pathways, whereas the structural connectivity of the neuronal layer modulates the competitive preference of individual nodes.
		}
		\label{fig:multilayer_network}
	\end{figure}
	
	As illustrated in Fig.~\ref{fig:multilayer_network}, the system is represented by two interconnected layers. The upper layer describes the electrical interactions among HR neurons, whereas the lower layer describes the redistribution of a finite resource. The two layers share the same set of neurons but have different edge structures. Let \(G_B\) denote the resource-transport network with adjacency matrix \(B=[b_{ij}]\). The transport layer provides possible pathways for resource redistribution, while the degree information of the neuronal interaction layer determines the competitive preference of individual neurons.
	
	To describe this redistribution process, a resource unit is regarded as a walker moving over \(G_B\). When the walker is located at node \(i\), the probability of transferring to a neighboring node \(j\) is defined as
	\begin{equation}\label{eq:transition_prob}
		\pi_{ij}
		=
		\frac{b_{ij}f_j}
		{\displaystyle\sum_{l=1}^{N}b_{il}f_l},
	\end{equation}
	where the topology-dependent preference of node \(j\) is
	\begin{equation}\label{eq:bias_function}
		f_j=k_j^{\alpha}.
	\end{equation}
	Here, \(k_j\) is the degree of neuron \(j\) in the neuronal interaction network, and \(\alpha\) is the degree-preference exponent. Equation~\eqref{eq:transition_prob} therefore separates two factors in the resource allocation process: \(b_{ij}\) determines whether a resource unit can be transported from node \(i\) to node \(j\), whereas \(k_j^{\alpha}\) determines the relative competitive advantage of the destination node.
	
	Let \(p_i(t)\) denote the fraction of the available resource assigned to neuron \(i\) at discrete time \(t\), satisfying
	\[
	p_i(t)\geq 0,
	\qquad
	\sum_{i=1}^{N}p_i(t)=1.
	\]
	The redistribution of the resource is then governed by the Markov process
	\begin{equation}\label{eq:markov_update}
		p_j(t+1)
		=
		\sum_{i=1}^{N}\pi_{ij}p_i(t).
	\end{equation}
	The vector
	\(p(t)=[p_1(t),p_2(t),\dots,p_N(t)]^\top\)
	therefore characterizes the instantaneous distribution of the common resource over the network.
	
	The exponent \(\alpha\) controls how strongly network structure affects resource competition. When \(\alpha>0\), transitions toward highly connected neurons are favored, and hub nodes tend to acquire a larger fraction of the available resource. When \(\alpha=0\), the structural preference disappears and the allocation reduces to an unbiased random walk on \(G_B\). When \(\alpha<0\), weakly connected neurons are preferentially favored. Varying \(\alpha\) therefore determines how strongly degree heterogeneity is reflected in resource allocation.
	
	For an ergodic transport process, \(p(t)\) converges to a stationary distribution \(p^{*}\), which determines the long-term allocation of the finite resource among neurons. This stationary resource field will subsequently be incorporated into the adaptive synchronization mechanism, so that the control capability of each neuron is jointly determined by its synchronization demand and its topology-dependent access to the shared resource.

	\section{RESOURCE-WEIGHTED ADAPTIVE SYNCHRONIZATION}
	
	Weak diffusive coupling is generally insufficient to guarantee synchronization of a network of chaotic oscillators. For Hindmarsh--Rose neurons, transverse perturbations may be amplified by the intrinsic nonlinear dynamics, whereas diffusive coupling provides only a finite dissipative effect through the membrane-potential channel. Synchronization can therefore be viewed as a competition between nonlinear expansion and network-induced dissipation. In the present framework, additional feedback is introduced to compensate for the missing dissipation, while its node-wise strength is constrained by the topology-dependent resource distribution established in Section~2.
	
	We first characterize the asymptotic resource field generated by the resource-transport process. To distinguish the discrete resource dynamics from the continuous neuronal dynamics, the redistribution step is indexed by \(n\). Define the probability simplex
	\begin{equation}\label{eq:probability_simplex}
		\Delta
		=
		\left\{
		p\in\mathbb{R}_{+}^{N}
		\,\middle|\,
		\mathbf{1}^{\top}p=1
		\right\}.
	\end{equation}
	Let \(P=[\pi_{ij}]_{N\times N}\) denote the transition matrix of the
	topology-dependent resource-transport process defined in
	Section~2. The resource distribution
	\(p^{(n)}=[p_1^{(n)},\dots,p_N^{(n)}]^\top\in\Delta\) evolves according to
	\begin{equation}\label{eq:resource_markov_ch3}
		p^{(n+1)}
		=
		P^\top p^{(n)}.
	\end{equation}
	Since \(P\mathbf{1}=\mathbf{1}\), the total resource is conserved.
	The constraint \(\mathbf{1}^{\top}p=1\) represents a finite shared
	resource pool, so that increasing the resource share of one node
	necessarily reduces that available to the others.
	
	For notational convenience, define
	\[
	w_i=k_i^{\alpha},
	\qquad
	c_i=\sum_{l=1}^{N}b_{il}w_l,
	\qquad
	\mathcal{Z}=\sum_{m=1}^{N}w_m c_m.
	\]
	
	\begin{theorem}\label{thm:stationary_resource}
		Consider the topology-dependent resource-transport process
		\eqref{eq:resource_markov_ch3}, with transition matrix
		\(P=[\pi_{ij}]_{N\times N}\) defined in Section~2.
		Suppose that \(G_B\) is connected and undirected.
		Then the associated Markov chain is irreducible and admits a unique
		stationary distribution \(p^*\in\Delta\), whose components are
		explicitly given by
		\begin{equation}\label{eq:stationary_resource_explicit}
			p_i^{*}
			=
			\frac{w_i c_i}{\mathcal{Z}}
			=
			\frac{
				k_i^{\alpha}
				\displaystyle\sum_{l=1}^{N}b_{il}k_l^{\alpha}
			}{
				\displaystyle\sum_{m=1}^{N}
				k_m^{\alpha}
				\sum_{l=1}^{N}b_{ml}k_l^{\alpha}
			},
			\qquad
			i=1,\dots,N.
		\end{equation}
		In particular,
		\begin{equation}\label{eq:stationary_resource_property}
			P^\top p^{*}=p^{*},
			\qquad
			p_i^{*}>0,
			\qquad
			\mathbf{1}^{\top}p^{*}=1.
		\end{equation}
		If the Markov chain is additionally aperiodic, then
		\begin{equation}\label{eq:stationary_resource_limit}
			\lim_{n\rightarrow\infty}
			(P^\top)^n p^{(0)}
			=
			p^{*},
			\qquad
			\forall\,p^{(0)}\in\Delta.
		\end{equation}
	\end{theorem}
	
	\begin{proof}
		The proof proceeds in three steps.
		
		First, for every node \(i\),
		\begin{equation}
			\sum_{j=1}^{N}\pi_{ij}
			=
			\frac{
				\sum_{j=1}^{N}b_{ij}k_j^{\alpha}
			}{
				\sum_{l=1}^{N}b_{il}k_l^{\alpha}
			}
			=1.
		\end{equation}
		Hence \(P\) is row stochastic, and therefore
		\[
		\mathbf{1}^{\top}p^{(n+1)}
		=
		\mathbf{1}^{\top}P^\top p^{(n)}
		=
		\mathbf{1}^{\top}p^{(n)}.
		\]
		Thus, if \(p^{(0)}\in\Delta\), then \(p^{(n)}\in\Delta\) for all
		\(n\geq0\).
		
		Second, we establish irreducibility. Since \(G_B\) is connected and
		\(k_j^{\alpha}>0\), every edge of \(G_B\) generates a strictly positive
		transition probability. For any pair of nodes \(i\) and \(j\), there
		exists a finite path
		\[
		i=v_0\rightarrow v_1\rightarrow\cdots\rightarrow v_m=j
		\]
		such that
		\begin{equation}
			\pi_{v_0v_1}
			\pi_{v_1v_2}
			\cdots
			\pi_{v_{m-1}v_m}>0.
		\end{equation}
		Consequently,
		\[
		(P^m)_{ij}>0
		\]
		for some finite \(m\), which proves that the Markov chain is irreducible.
		
		Finally, we determine its stationary distribution. Using the above
		notation and \eqref{eq:stationary_resource_explicit},
		\begin{equation}
			\begin{aligned}
				p_i^{*}\pi_{ij}
				&=
				\frac{w_i c_i}{\mathcal Z}
				\frac{b_{ij}w_j}{c_i}\\
				&=
				\frac{b_{ij}w_iw_j}{\mathcal Z}.
			\end{aligned}
		\end{equation}
		Because \(G_B\) is undirected, \(b_{ij}=b_{ji}\), and hence
		\begin{equation}\label{eq:detailed_balance}
			p_i^{*}\pi_{ij}
			=
			p_j^{*}\pi_{ji}.
		\end{equation}
		Thus, \(p^*\) satisfies the detailed-balance relation. Summing
		\eqref{eq:detailed_balance} over \(i\) gives
		\begin{equation}
			\begin{aligned}
				(P^\top p^*)_j
				&=
				\sum_{i=1}^{N}\pi_{ij}p_i^*\\
				&=
				p_j^*\sum_{i=1}^{N}\pi_{ji}\\
				&=
				p_j^*,
			\end{aligned}
		\end{equation}
		and therefore \(P^\top p^*=p^*\).
		
		Since the chain is finite and irreducible, the stationary distribution
		is unique and strictly positive. If it is additionally aperiodic, the
		transition matrix is primitive. The Perron--Frobenius theorem then gives
		\[
		(P^\top)^n p^{(0)}
		\rightarrow p^*
		\]
		for every \(p^{(0)}\in\Delta\), completing the proof.
	\end{proof}
	
	Theorem~\ref{thm:stationary_resource} gives a direct relation between network structure and long-term resource availability:
	\begin{equation}\label{eq:stationary_resource_structure}
		p_i^{*}
		\propto
		k_i^{\alpha}
		\sum_{l=1}^{N}b_{il}k_l^{\alpha}.
	\end{equation}
	The first factor describes the structural preference of node \(i\), whereas the second reflects its weighted accessibility through the transport layer. Hence, both the neuronal topology and the transport topology contribute to the stationary resource heterogeneity. In particular, \(\alpha\) controls how strongly structural centrality is amplified or suppressed in the stationary resource field.
	
	The existence of \(p^*\) alone is not sufficient for coupling it to the continuous neuronal dynamics. The rate at which the resource field relaxes toward \(p^*\) must also be characterized.
	
	\begin{theorem}\label{thm:resource_exponential}
		Under the conditions of Theorem~\ref{thm:stationary_resource}, suppose
		further that the Markov chain is aperiodic. Define
		\begin{equation}
			W_*
			=
			\operatorname{diag}
			\left(
			p_1^*,p_2^*,\dots,p_N^*
			\right)
		\end{equation}
		and the weighted norm
		\begin{equation}\label{eq:weighted_resource_norm}
			\|z\|_{W_*^{-1}}
			=
			\sqrt{
				z^\top W_*^{-1}z
			}.
		\end{equation}
		Let
		\[
		1=\lambda_1,\lambda_2,\dots,\lambda_N
		\]
		denote the eigenvalues of \(P\), and define
		\begin{equation}\label{eq:resource_gamma}
			\gamma
			=
			\max_{2\leq r\leq N}
			|\lambda_r|.
		\end{equation}
		Then \(0\leq\gamma<1\), and for any \(p^{(0)}\in\Delta\),
		\begin{equation}\label{eq:weighted_exponential_convergence}
			\left\|
			p^{(n)}-p^*
			\right\|_{W_*^{-1}}
			\leq
			\gamma^n
			\left\|
			p^{(0)}-p^*
			\right\|_{W_*^{-1}}.
		\end{equation}
		Consequently,
		\begin{equation}\label{eq:euclidean_exponential_convergence}
			\left\|
			p^{(n)}-p^*
			\right\|_2
			\leq
			C_R\gamma^n
			\left\|
			p^{(0)}-p^*
			\right\|_2,
		\end{equation}
		where
		\begin{equation}\label{eq:resource_norm_constant}
			C_R
			=
			\sqrt{
				\frac{p_{\max}^*}{p_{\min}^*}
			},
			\qquad
			p_{\max}^*=\max_i p_i^*,
			\qquad
			p_{\min}^*=\min_i p_i^*.
		\end{equation}
		Therefore, the topology-dependent resource field converges exponentially
		to \(p^*\).
	\end{theorem}
	
	\begin{proof}
		The detailed-balance relation \eqref{eq:detailed_balance} is equivalent to
		\begin{equation}\label{eq:detailed_balance_matrix}
			W_*P
			=
			P^\top W_*.
		\end{equation}
		Consider the matrix
		\begin{equation}\label{eq:symmetric_markov_operator}
			S
			=
			W_*^{1/2}
			P
			W_*^{-1/2}.
		\end{equation}
		Using \eqref{eq:detailed_balance_matrix},
		\begin{equation}
			\begin{aligned}
				S^\top
				&=
				W_*^{-1/2}
				P^\top
				W_*^{1/2}\\
				&=
				W_*^{1/2}
				P
				W_*^{-1/2}
				=
				S.
			\end{aligned}
		\end{equation}
		Thus \(S\) is symmetric and similar to \(P\). Hence all eigenvalues of
		\(P\) are real and \(P\) is diagonalizable.
		
		Since the Markov chain is irreducible and aperiodic, the Perron
		eigenvalue \(\lambda_1=1\) is simple, whereas
		\[
		|\lambda_r|<1,
		\qquad
		r=2,\dots,N.
		\]
		Therefore,
		\[
		0\leq\gamma<1.
		\]
		
		Let
		\begin{equation}
			z^{(n)}
			=
			p^{(n)}-p^*
		\end{equation}
		and introduce
		\begin{equation}
			y^{(n)}
			=
			W_*^{-1/2}z^{(n)}.
		\end{equation}
		Because \(P^\top p^*=p^*\),
		\[
		z^{(n+1)}
		=
		P^\top z^{(n)},
		\]
		and therefore
		\begin{equation}
			\begin{aligned}
				y^{(n+1)}
				&=
				W_*^{-1/2}
				P^\top
				W_*^{1/2}
				y^{(n)}\\
				&=
				S y^{(n)}.
			\end{aligned}
		\end{equation}
		
		The normalized stationary direction of \(S\) is proportional to
		\[
		q_1=W_*^{1/2}\mathbf{1}.
		\]
		Furthermore,
		\begin{equation}
			\begin{aligned}
				q_1^\top y^{(n)}
				&=
				\mathbf{1}^\top
				z^{(n)}\\
				&=
				\mathbf{1}^\top
				\bigl(p^{(n)}-p^*\bigr)
				=0.
			\end{aligned}
		\end{equation}
		Hence \(y^{(n)}\) lies entirely in the subspace transverse to the
		stationary mode. Since \(S\) is symmetric, its induced \(2\)-norm on
		this subspace is exactly \(\gamma\). Thus,
		\begin{equation}
			\|y^{(n)}\|_2
			\leq
			\gamma^n
			\|y^{(0)}\|_2.
		\end{equation}
		By definition,
		\[
		\|y^{(n)}\|_2
		=
		\|p^{(n)}-p^*\|_{W_*^{-1}},
		\]
		which proves \eqref{eq:weighted_exponential_convergence}.
		
		Finally, for any \(z\in\mathbb{R}^N\),
		\begin{equation}
			\frac{1}{p_{\max}^*}\|z\|_2^2
			\leq
			\|z\|_{W_*^{-1}}^2
			\leq
			\frac{1}{p_{\min}^*}\|z\|_2^2.
		\end{equation}
		Combining this norm equivalence with
		\eqref{eq:weighted_exponential_convergence} gives
		\begin{equation}
			\|p^{(n)}-p^*\|_2
			\leq
			\sqrt{
				\frac{p_{\max}^*}{p_{\min}^*}
			}
			\gamma^n
			\|p^{(0)}-p^*\|_2,
		\end{equation}
		which proves \eqref{eq:euclidean_exponential_convergence}.
	\end{proof}
	
	Theorem~\ref{thm:resource_exponential} shows that the topology-dependent resource field converges exponentially to its unique stationary distribution \(p^*\). If one redistribution step corresponds to a physical interval \(\Delta t_R\), the dominant relaxation time can be characterized by
	\begin{equation}\label{eq:resource_relaxation_time}
		T_R
		=
		-\frac{\Delta t_R}{\ln\gamma},
		\qquad
		0<\gamma<1.
	\end{equation}
	Let \(T_X\) denote a characteristic time scale of the neuronal dynamics. We consider the separated-timescale regime
	\begin{equation}\label{eq:qssa_condition}
		\varepsilon_R
		=
		\frac{T_R}{T_X}
		\ll1.
	\end{equation}
	Under this condition, the resource redistribution relaxes substantially faster than the neuronal dynamics, and hence
	\begin{equation}\label{eq:qssa_resource}
		p^{(n)}
		\simeq
		p^*
	\end{equation}
	is adopted on the neuronal time scale. The resource-transport layer and the neuronal layer evolve concurrently in the full duplex dynamics. The quasi-steady approximation \eqref{eq:qssa_resource} is used only in the analytical treatment, based on the faster relaxation of the resource process, and does not imply sequential evolution of the two layers. The quantity \(p_i^*>0\) therefore acts as a constant topology-dependent measure of the resource available to neuron \(i\) in the subsequent stability analysis.
	
	We next introduce a fully local synchronization signal. Since both the diffusive interaction and the feedback controller act through the membrane-potential channel, define
	\begin{equation}\label{eq:local_error_scalar}
		\eta_i(t)
		=
		x_i(t)
		-
		\frac{1}{k_i}
		\sum_{j\in\mathcal N_i}
		a_{ij}x_j(t)
		=
		\frac{1}{k_i}
		\sum_{j=1}^{N}
		l_{ij}x_j(t).
	\end{equation}
	Thus, \(\eta_i(t)\) measures the deviation of neuron \(i\) from the instantaneous mean state of its neighbors and requires no global network information.
	
	Let
	\[
	x=[x_1,\dots,x_N]^\top,
	\qquad
	e_x=x-\bar x\mathbf{1},
	\qquad
	D=\operatorname{diag}(k_1,\dots,k_N),
	\]
	and
	\[
	\eta=[\eta_1,\dots,\eta_N]^\top .
	\]
	Then
	\begin{equation}\label{eq:local_error_matrix}
		\eta
		=
		D^{-1}Lx
		=
		D^{-1}Le_x,
	\end{equation}
	where the second equality follows from \(L\mathbf{1}=0\).
	For a connected undirected graph,
	\begin{equation}\label{eq:local_global_equivalence}
		\frac{\lambda_2(L)}{k_{\max}}
		\|e_x\|_2
		\leq
		\|\eta\|_2
		\leq
		\frac{\lambda_N(L)}{k_{\min}}
		\|e_x\|_2,
	\end{equation}
	where
	\[
	k_{\min}=\min_i k_i,
	\qquad
	k_{\max}=\max_i k_i.
	\]
	Therefore,
	\begin{equation}\label{eq:local_sync_equivalence}
		\eta(t)\rightarrow0
		\quad\Longleftrightarrow\quad
		e_x(t)\rightarrow0.
	\end{equation}
	Equation~\eqref{eq:local_sync_equivalence} establishes an exact equivalence between the local disagreement coordinate \(\eta\) and the mean-field synchronization error. The feedback can therefore be implemented without global mean-field information.
	
	Based on the local disagreement and the stationary resource field, we introduce the resource-weighted adaptive feedback law
	\begin{equation}\label{eq:local_controller}
		u_i(t)
		=
		-d_i(t)\eta_i(t),
		\qquad
		d_i(0)\geq0,
	\end{equation}
	with the adaptive gain
	\begin{equation}\label{eq:local_adaptation}
		\dot d_i(t)
		=
		\delta_i p_i^*
		\eta_i^2(t)e^{\mu t},
		\qquad
		\delta_i>0,
		\qquad
		\mu>0.
	\end{equation}
	The adaptation law has a direct local interpretation: the feedback gain increases only when a node simultaneously possesses available resource \(p_i^*\) and exhibits nonzero synchronization disagreement \(\eta_i\). Moreover, since \(\dot d_i\geq0\), the gain remains nonnegative for all \(t\geq0\).
	
	Before establishing synchronization, we first verify that the closed-loop dynamics remain well posed and bounded.
	
	\begin{lemma}\label{lem:closed_loop_boundedness}
		Consider the closed-loop HR network
		\eqref{eq:net_general} under the local feedback law
		\eqref{eq:local_controller}--\eqref{eq:local_adaptation}.
		Suppose that \(G\) is connected and undirected, and that
		\(a>0\), \(d_0>0\), \(r>0\), and \(s_0>0\).
		Then, for every finite initial condition with \(d_i(0)\geq0\),
		the closed-loop solution is forward complete, and the neuronal
		states
		\[
		X_i(t)=[x_i(t),y_i(t),z_i(t)]^\top
		\]
		remain bounded for all \(t\geq0\).
	\end{lemma}
	
	\begin{proof}
		From \eqref{eq:local_adaptation},
		\[
		\dot d_i(t)
		=
		\delta_i p_i^*\eta_i^2(t)e^{\mu t}
		\geq0,
		\]
		and therefore
		\[
		d_i(t)\geq d_i(0)\geq0
		\]
		throughout the interval of existence.
		
		Using
		\[
		\sum_{j=1}^{N}l_{ij}x_j
		=
		k_i\eta_i,
		\]
		the membrane-potential equation of the closed-loop network can be
		written as
		\begin{equation}\label{eq:boundedness_x_dynamics}
			\dot x_i
			=
			y_i-a x_i^3+b x_i^2-z_i+I_{\rm ext}
			-
			\bigl(\sigma k_i+d_i\bigr)\eta_i .
		\end{equation}
		
		We construct a bounded positively invariant region for the neuronal
		states. Choose \(R>\max_{1\leq i\leq N}|x_i(0)|\) sufficiently large
		and define
		\begin{equation}\label{eq:bounded_YR}
			Y_R
			=
			\max
			\left\{
			\max_{1\leq i\leq N}|y_i(0)|+1,\,
			d_0R^2+|c|+1
			\right\},
		\end{equation}
		and
		\begin{equation}\label{eq:bounded_ZR}
			Z_R
			=
			\max
			\left\{
			\max_{1\leq i\leq N}|z_i(0)|+1,\,
			s_0(R+|x_0|)+1
			\right\}.
		\end{equation}
		Consider the compact set
		\begin{equation}\label{eq:bounded_region}
			\Omega_R
			=
			[-R,R]^N
			\times
			[-Y_R,Y_R]^N
			\times
			[-Z_R,Z_R]^N .
		\end{equation}
		
		For the fast recovery variable,
		\[
		\dot y_i
		=
		c-d_0x_i^2-y_i.
		\]
		On the boundary \(y_i=Y_R\),
		\[
		\dot y_i
		=
		c-d_0x_i^2-Y_R
		\leq
		|c|-Y_R
		<0,
		\]
		whereas on \(y_i=-Y_R\),
		\[
		\dot y_i
		=
		c-d_0x_i^2+Y_R
		\geq
		-|c|-d_0R^2+Y_R
		>0.
		\]
		Hence the vector field points inward on the two \(y\)-boundaries.
		
		For the slow recovery variable,
		\[
		\dot z_i
		=
		r\left[s_0(x_i+x_0)-z_i\right].
		\]
		On \(z_i=Z_R\),
		\[
		\dot z_i
		\leq
		r\left[s_0(R+|x_0|)-Z_R\right]
		<0,
		\]
		while on \(z_i=-Z_R\),
		\[
		\dot z_i
		\geq
		r\left[-s_0(R+|x_0|)+Z_R\right]
		>0.
		\]
		Thus the vector field also points inward on the \(z\)-boundaries.
		
		It remains to examine the membrane-potential boundaries. Suppose
		first that \(x_i=R\). Since \(x_j\leq R\) for every neighboring
		node \(j\),
		\[
		\eta_i
		=
		R-
		\frac{1}{k_i}
		\sum_{j\in\mathcal N_i}a_{ij}x_j
		\geq0.
		\]
		Because \(\sigma k_i+d_i\geq0\), the network coupling and adaptive
		feedback in \eqref{eq:boundedness_x_dynamics} can only point inward
		at this upper boundary. Therefore,
		\begin{equation}\label{eq:x_upper_bound_direction}
			\dot x_i
			\leq
			Y_R-aR^3+|b|R^2+Z_R+|I_{\rm ext}|.
		\end{equation}
		
		Likewise, if \(x_i=-R\), then every neighboring state satisfies
		\(x_j\geq-R\), and hence
		\[
		\eta_i
		=
		-R-
		\frac{1}{k_i}
		\sum_{j\in\mathcal N_i}a_{ij}x_j
		\leq0.
		\]
		The coupling-feedback contribution is now nonnegative, and
		\begin{equation}\label{eq:x_lower_bound_direction}
			\dot x_i
			\geq
			-Y_R+aR^3-|b|R^2-Z_R-|I_{\rm ext}|.
		\end{equation}
		
		By construction,
		\[
		Y_R=O(R^2),
		\qquad
		Z_R=O(R),
		\]
		whereas the leading restoring term of the HR dynamics is cubic.
		Since \(a>0\), \(R\) can therefore be chosen sufficiently large such
		that
		\begin{equation}\label{eq:R_bounded_condition}
			aR^3
			>
			Y_R+|b|R^2+Z_R+|I_{\rm ext}|.
		\end{equation}
		For such an \(R\),
		\[
		x_i=R
		\quad\Longrightarrow\quad
		\dot x_i<0,
		\]
		and
		\[
		x_i=-R
		\quad\Longrightarrow\quad
		\dot x_i>0.
		\]
		Hence the vector field points inward on every boundary of
		\(\Omega_R\), and the neuronal trajectory cannot leave this compact
		set.
		
		Finally, within \(\Omega_R\),
		\[
		|\eta_i(t)|
		\leq
		2R.
		\]
		Consequently, for every finite \(T>0\),
		\begin{align}
			d_i(t)
			&=
			d_i(0)
			+
			\delta_i p_i^*
			\int_0^t
			\eta_i^2(s)e^{\mu s}\,ds \nonumber\\
			&\leq
			d_i(0)
			+
			4\delta_i p_i^*R^2
			\int_0^T e^{\mu s}\,ds
			<\infty,
			\qquad
			0\leq t\leq T.
		\end{align}
		Thus the adaptive gains cannot undergo finite-time escape. Since the
		right-hand side of the augmented closed-loop system is locally
		Lipschitz and all state components remain finite on every bounded
		time interval, the solution can be continued for arbitrary
		\(t\geq0\). Therefore the closed-loop system is forward complete,
		and all neuronal states remain bounded.
	\end{proof}
	
	\begin{theorem}\label{thm:local_adaptive_sync}
		Consider the HR network \eqref{eq:net_general} over a connected
		undirected graph \(G\), and let
		\(p^*\in\mathbb{R}_{>0}^N\) be the stationary resource distribution
		given by Theorems~\ref{thm:stationary_resource} and
		\ref{thm:resource_exponential} under the quasi-steady-state
		approximation. Under the conditions of
		Lemma~\ref{lem:closed_loop_boundedness}, for any
		\(\delta_i>0\), \(\mu>0\), and \(d_i(0)\geq0\), the local adaptive
		feedback law
		\eqref{eq:local_controller}--\eqref{eq:local_adaptation}
		drives the network to the synchronization manifold, i.e.,
		\begin{equation}\label{eq:thm3_sync_result}
			\lim_{t\to\infty}
			\|X_i(t)-X_j(t)\|_2
			=
			0,
			\qquad
			\forall\,i,j.
		\end{equation}
		Equivalently,
		\begin{equation}\label{eq:thm3_mean_sync_result}
			\lim_{t\to\infty}
			\|X_i(t)-\bar X(t)\|_2
			=
			0,
			\qquad
			i=1,\dots,N.
		\end{equation}
	\end{theorem}
	
	\begin{proof}
		Let
		\[
		x=[x_1,\dots,x_N]^\top,\qquad
		y=[y_1,\dots,y_N]^\top,\qquad
		z=[z_1,\dots,z_N]^\top,
		\]
		and denote
		\[
		x^{[2]}
		=
		[x_1^2,\dots,x_N^2]^\top,
		\qquad
		x^{[3]}
		=
		[x_1^3,\dots,x_N^3]^\top.
		\]
		
		By Lemma~\ref{lem:closed_loop_boundedness}, the neuronal trajectories
		remain in a compact positively invariant set. Hence there exists a
		constant \(B>0\) such that
		\begin{equation}\label{eq:x_bounded_thm3}
			|x_i(t)|
			\leq B,
			\qquad
			i=1,\dots,N,
			\qquad
			t\geq0.
		\end{equation}
		
		Choose
		\begin{equation}\label{eq:kappa_thm3}
			\kappa
			=
			\frac{1}{rs_0}>0
		\end{equation}
		and introduce the graph-weighted synchronization measure
		\begin{equation}\label{eq:graph_sync_measure}
			V_s
			=
			\frac{1}{2}
			\left(
			x^\top Lx
			+
			y^\top Ly
			+
			\kappa z^\top Lz
			\right).
		\end{equation}
		Because \(G\) is connected,
		\[
		V_s=0
		\quad\Longleftrightarrow\quad
		X_1=X_2=\cdots=X_N.
		\]
		Thus, \(V_s\) measures the distance of the network from the
		synchronization manifold without introducing the global mean into
		the controller.
		
		Let \(d^\star>0\) be an auxiliary constant to be specified below and
		define
		\begin{equation}\label{eq:adaptive_lyapunov_term}
			V_d
			=
			\frac{1}{2}
			\sum_{i=1}^{N}
			\frac{
				k_i(d_i-d^\star)^2
			}{
				\delta_i p_i^*
			}
			e^{-\mu t}.
		\end{equation}
		Consider the Lyapunov function
		\begin{equation}\label{eq:local_lyapunov}
			V
			=
			V_s+V_d.
		\end{equation}

		Differentiating the synchronization part gives
		\begin{equation}\label{eq:Vs_dot_initial}
			\dot V_s
			=
			x^\top L\dot x
			+
			y^\top L\dot y
			+
			\kappa z^\top L\dot z.
		\end{equation}
		Substituting the HR dynamics and using \(L\mathbf{1}=0\), all
		constant terms vanish, yielding
		\begin{equation}\label{eq:Vs_dot_expand}
			\begin{aligned}
				\dot V_s
				&=
				x^\top Ly
				-a x^\top Lx^{[3]}
				+b x^\top Lx^{[2]}
				-x^\top Lz\\
				&\quad
				-\sigma x^\top L^2x
				+x^\top Lu\\
				&\quad
				-d_0y^\top Lx^{[2]}
				-y^\top Ly\\
				&\quad
				+\kappa rs_0 z^\top Lx
				-\kappa r z^\top Lz .
			\end{aligned}
		\end{equation}
		Since \(L=L^\top\) and \(\kappa rs_0=1\),
		\begin{equation}\label{eq:xz_cancel}
			-x^\top Lz
			+
			\kappa rs_0z^\top Lx
			=
			0.
		\end{equation}
		
		We now estimate the remaining nonlinear terms. For the cubic HR
		nonlinearity,
		\begin{equation}\label{eq:cubic_graph_identity}
			\begin{aligned}
				x^\top Lx^{[3]}
				&=
				\frac{1}{2}
				\sum_{i,j=1}^{N}
				a_{ij}
				(x_i-x_j)
				(x_i^3-x_j^3)\\
				&=
				\frac{1}{2}
				\sum_{i,j=1}^{N}
				a_{ij}
				(x_i-x_j)^2
				\left(
				x_i^2+x_ix_j+x_j^2
				\right)
				\geq0.
			\end{aligned}
		\end{equation}
		Since the HR parameter \(a>0\),
		\begin{equation}\label{eq:cubic_dissipative}
			-a x^\top Lx^{[3]}
			\leq0.
		\end{equation}
		Thus, the cubic nonlinearity is intrinsically dissipative in the
		graph-disagreement coordinates and need not be bounded by a
		Lipschitz constant.
		
		For the quadratic term, \eqref{eq:x_bounded_thm3} gives
		\[
		|x_i^2-x_j^2|
		\leq
		2B|x_i-x_j|,
		\]
		and hence
		\begin{equation}\label{eq:quadratic_graph_bound}
			\left|
			x^\top Lx^{[2]}
			\right|
			\leq
			2B\,x^\top Lx.
		\end{equation}
		Furthermore, by the Cauchy--Schwarz and Young inequalities,
		\begin{equation}\label{eq:xy_graph_bound}
			x^\top Ly
			\leq
			\frac{1}{2}x^\top Lx
			+
			\frac{1}{2}y^\top Ly.
		\end{equation}
		Similarly,
		\begin{equation}\label{eq:y_x2_graph_bound}
			\left|
			y^\top Lx^{[2]}
			\right|
			\leq
			2B
			\sqrt{
				(x^\top Lx)(y^\top Ly)
			},
		\end{equation}
		and therefore
		\begin{equation}\label{eq:y_x2_young}
			\left|
			d_0y^\top Lx^{[2]}
			\right|
			\leq
			\frac{1}{4}y^\top Ly
			+
			4d_0^2B^2x^\top Lx.
		\end{equation}
		
		Define
		\begin{equation}\label{eq:CB_definition}
			C_B
			=
			\frac{1}{2}
			+
			2|b|B
			+
			4d_0^2B^2.
		\end{equation}
		Using
		\eqref{eq:cubic_dissipative}--\eqref{eq:y_x2_young}
		in \eqref{eq:Vs_dot_expand}, we obtain
		\begin{equation}\label{eq:Vs_bound_before_control}
			\begin{aligned}
				\dot V_s
				&\leq
				C_B x^\top Lx
				-\frac{1}{4}y^\top Ly
				-\kappa r z^\top Lz\\
				&\quad
				-\sigma\|Lx\|_2^2
				+
				(Lx)^\top u .
			\end{aligned}
		\end{equation}

		From \eqref{eq:local_error_scalar},
		\begin{equation}\label{eq:Lx_eta_relation}
			(Lx)_i
			=
			k_i\eta_i.
		\end{equation}
		Therefore,
		\begin{equation}\label{eq:local_control_dissipation}
			\begin{aligned}
				(Lx)^\top u
				&=
				\sum_{i=1}^{N}
				(Lx)_iu_i\\
				&=
				-\sum_{i=1}^{N}
				k_i d_i\eta_i^2.
			\end{aligned}
		\end{equation}
		
		Differentiating \(V_d\) and using
		\eqref{eq:local_adaptation} gives
		\begin{equation}\label{eq:Vd_dot}
			\begin{aligned}
				\dot V_d
				&=
				\sum_{i=1}^{N}
				\frac{
					k_i(d_i-d^\star)\dot d_i
				}{
					\delta_i p_i^*
				}
				e^{-\mu t}
				-\mu V_d\\
				&=
				\sum_{i=1}^{N}
				k_i(d_i-d^\star)\eta_i^2
				-\mu V_d.
			\end{aligned}
		\end{equation}
		Combining \eqref{eq:local_control_dissipation} and
		\eqref{eq:Vd_dot} yields the exact cancellation
		\begin{equation}\label{eq:exact_local_cancellation}
			(Lx)^\top u+\dot V_d
			=
			-d^\star
			\sum_{i=1}^{N}
			k_i\eta_i^2
			-\mu V_d.
		\end{equation}

		Moreover,
		\begin{equation}\label{eq:eta_quadratic_form}
			\sum_{i=1}^{N}
			k_i\eta_i^2
			=
			(Lx)^\top D^{-1}(Lx).
		\end{equation}
		Since \(D\preceq k_{\max}I_N\),
		\begin{equation}\label{eq:eta_lower_1}
			(Lx)^\top D^{-1}(Lx)
			\geq
			\frac{1}{k_{\max}}
			\|Lx\|_2^2.
		\end{equation}
		For a connected graph,
		\begin{equation}\label{eq:L2_L_relation}
			\|Lx\|_2^2
			=
			x^\top L^2x
			\geq
			\lambda_2(L)x^\top Lx.
		\end{equation}
		Consequently,
		\begin{equation}\label{eq:eta_lower_graph}
			\sum_{i=1}^{N}
			k_i\eta_i^2
			\geq
			\frac{\lambda_2(L)}{k_{\max}}
			x^\top Lx.
		\end{equation}
		
		Combining
		\eqref{eq:Vs_bound_before_control},
		\eqref{eq:exact_local_cancellation}, and
		\eqref{eq:eta_lower_graph} gives
		\begin{equation}\label{eq:Vdot_local_pre}
			\begin{aligned}
				\dot V
				&\leq
				\left[
				C_B
				-\sigma\lambda_2(L)
				-\frac{d^\star\lambda_2(L)}{k_{\max}}
				\right]
				x^\top Lx\\
				&\quad
				-\frac{1}{4}y^\top Ly
				-\kappa r z^\top Lz
				-\mu V_d.
			\end{aligned}
		\end{equation}
		
		The constant \(d^\star\) is free and appears only in the Lyapunov
		function. Hence it can always be chosen sufficiently large such that
		\begin{equation}\label{eq:auxiliary_cx}
			c_x
			=
			\frac{d^\star\lambda_2(L)}{k_{\max}}
			+\sigma\lambda_2(L)
			-C_B
			>0.
		\end{equation}
		This choice imposes no constraint on the actual feedback gains
		\(d_i(t)\) or on their initial values. Therefore,
		\begin{equation}\label{eq:Vdot_local_negative}
			\dot V
			\leq
			-c_x x^\top Lx
			-\frac{1}{4}y^\top Ly
			-\kappa r z^\top Lz
			-\mu V_d.
		\end{equation}
		
		Let
		\begin{equation}\label{eq:rho_thm3}
			\rho
			=
			\min
			\left\{
			2c_x,\,
			\frac{1}{2},\,
			2r,\,
			\mu
			\right\}
			>0.
		\end{equation}
		From the definition \(V=V_s+V_d\), it follows that
		\begin{equation}\label{eq:Vdot_exponential_local}
			\dot V
			\leq
			-\rho V,
		\end{equation}
		and hence
		\begin{equation}\label{eq:V_exponential_local}
			V(t)
			\leq
			V(0)e^{-\rho t}.
		\end{equation}
		
		It remains to relate the graph disagreement measure to the usual
		synchronization error. Define
		\[
		e_x=x-\bar x\mathbf{1},
		\qquad
		e_y=y-\bar y\mathbf{1},
		\qquad
		e_z=z-\bar z\mathbf{1}.
		\]
		Because \(L\mathbf{1}=0\) and the graph is connected,
		\begin{equation}\label{eq:graph_mean_error_bound}
			\begin{aligned}
				x^\top Lx&\geq\lambda_2(L)\|e_x\|_2^2,\\
				y^\top Ly&\geq\lambda_2(L)\|e_y\|_2^2,\\
				z^\top Lz&\geq\lambda_2(L)\|e_z\|_2^2.
			\end{aligned}
		\end{equation}
		Using \eqref{eq:V_exponential_local},
		\begin{equation}\label{eq:full_error_bound_local}
			\sum_{i=1}^{N}
			\|X_i(t)-\bar X(t)\|_2^2
			\leq
			\frac{
				2\max\{1,\kappa^{-1}\}
			}{
				\lambda_2(L)
			}
			V(0)e^{-\rho t}.
		\end{equation}
		Thus,
		\[
		X_i(t)-\bar X(t)\rightarrow0,
		\qquad
		i=1,\dots,N,
		\]
		which is equivalent to
		\[
		X_i(t)-X_j(t)\rightarrow0,
		\qquad
		\forall\,i,j.
		\]
		Hence the controlled HR network achieves synchronization.
	\end{proof}
	
	\begin{figure}[!b]
		\centering
		\includegraphics[width=\linewidth]{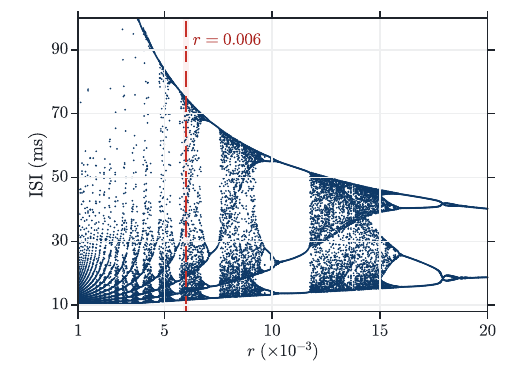}
		\caption{
			Interspike-interval bifurcation diagram of the HR oscillator with respect to \(r\) at \(I_{\mathrm{ext}}=3.2\). Alternating chaotic and periodic bursting windows identify the principal regimes of the intrinsic neuronal dynamics.
		}
		\label{fig:hr_bifurcation}
	\end{figure}
	
	\section{Experiments and Results}
	
	The experiments examine the formation of the topology-dependent stationary resource field and the synchronization dynamics generated by the resulting resource-weighted adaptive feedback. The neuronal interaction layer \(G\) and the resource-transport layer \(G_B\) share \(N=1000\) nodes and are generated independently. Their structural properties are summarized in Table~\ref{tab:network_statistics}.
	
	\begin{table}[t]
		\centering
		\caption{Structural properties of the two network layers.}
		\label{tab:network_statistics}
		\small
		\setlength{\tabcolsep}{5pt}
		\renewcommand{\arraystretch}{1.12}
		\begin{tabular}{@{}lcc@{}}
			\toprule
			Property
			& \shortstack{Neuronal\\layer \(G\)}
			& \shortstack{Transport\\layer \(G_B\)} \\
			\midrule
			Topology & Barabási--Albert & Erdős--Rényi \\
			Generation parameter & \(m=3\) & \(q_B=0.1\) \\
			Degree range & -- & \([73,136]\) \\
			Mean degree & \(5.984\) & \(99.68\) \\
			Clustering coefficient & \(0.0341\) & \(0.0998\) \\
			Average path length & \(3.4493\) & \(1.8985\) \\
			\bottomrule
		\end{tabular}
	\end{table}
	
	The selected realization of \(G_B\) is connected, and a self-loop \(b_{ii}=1\) is added to every node, making the associated Markov chain irreducible and aperiodic. The HR parameters are
	\(a=1\), \(b=3\), \(c=1\), \(d_0=5\), \(s_0=4\), \(x_0=1.6\), and \(I_{\mathrm{ext}}=3.2\). The slow time-scale parameter is fixed at \(r=0.006\), selected from the chaotic regime in Fig.~\ref{fig:hr_bifurcation}. The adaptive parameters are \(\delta_i=0.2\), \(\mu=0.35\), and \(d_i(0)=0\). Unless otherwise specified, \(\sigma=0.05\). The continuous dynamics are integrated by fourth-order Runge--Kutta with \(h=0.01\) over \(t\in[0,100]\). 
	
	Figure~\ref{fig:hr_bifurcation} shows the repeated alternation between periodic and chaotic bursting as \(r\) varies. The value \(r=0.006\) lies in a chaotic bursting window and is used in all subsequent experiments.
	
	\subsection{Formation of the Stationary Resource Field}
	\label{subsec:stationary_resource_field}
	
	\begin{figure*}[!t]
		\centering
		\begin{minipage}[t]{0.40\textwidth}
			\centering
			\makebox[\linewidth][l]{\textbf{(a)}}\par
			\vspace{0.4ex}
			\includegraphics[width=\linewidth]{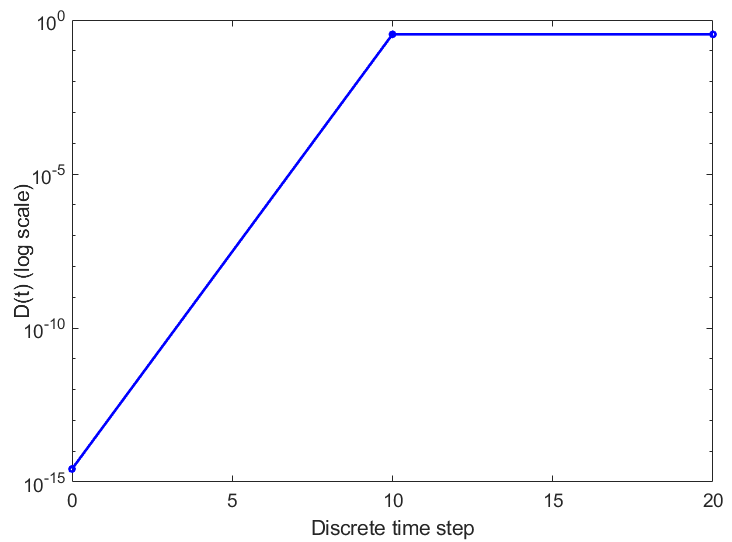}
		\end{minipage}
		\hfill
		\begin{minipage}[t]{0.55\textwidth}
			\centering
			\makebox[\linewidth][l]{\textbf{(b)}}\par
			\vspace{0.4ex}
			\includegraphics[width=\linewidth]{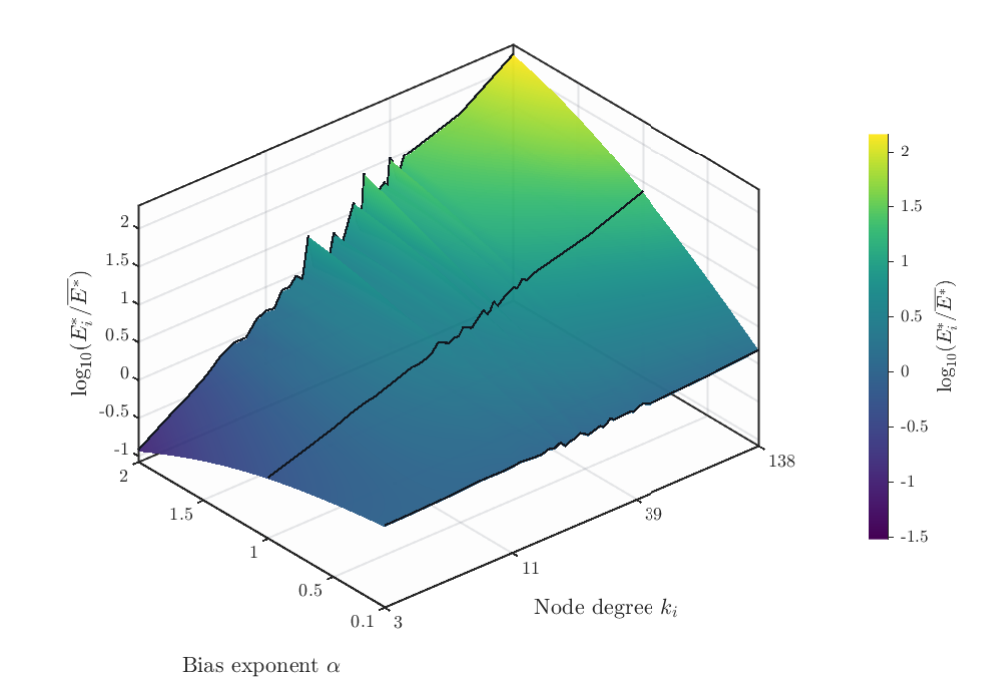}
		\end{minipage}
		\caption{Formation of the stationary resource field.
			\textbf{(a)} Departure \(D(n)=\|p^{(n)}-p^{(0)}\|_1\) from the initial allocation.
			\textbf{(b)} Degree-conditioned stationary resource landscape over \(\alpha\in[0.1,2]\); the height and color represent \(\log_{10}\overline{\varepsilon}(k,\alpha)\), where \(\varepsilon_i=Np_i^*\).}
		\label{fig:stationary_resource_field}
	\end{figure*}
	
	The resource process reaches its stationary field rapidly. We monitor the redistribution by
	\begin{equation}
		D(n)=\left\|p^{(n)}-p^{(0)}\right\|_1 .
		\label{eq:departure_from_initial}
	\end{equation}
	Figure~\ref{fig:stationary_resource_field}(a) shows that \(D(n)\) reaches a plateau after approximately ten transport steps, consistent with the fast relaxation established in Theorem~3.2. This transient is short compared with the HR evolution considered below and supports the quasi-stationary treatment of the resource field.
	
	To display the stationary allocation on a common scale, we use the relative resource \(\varepsilon_i(\alpha)=Np_i^*(\alpha)\); \(\varepsilon_i=1\) corresponds to the network-average allocation. For nodes of degree \(k\), \(\overline{\varepsilon}(k,\alpha)\) denotes the degree-conditioned mean. Figure~\ref{fig:stationary_resource_field}(b) shows a progressive tilt of the stationary field toward high-degree nodes as \(\alpha\) increases. At \(\alpha=0.1\), the allocation remains broadly distributed. At \(\alpha=1\), the highest-degree \(10\%\) of nodes hold about one third of the total resource, whereas at \(\alpha=2\) their share rises to roughly three quarters. The total resource is unchanged; the bias alters where it is available.
	
	Because \(p_i^*\) directly weights the gain growth in Eq.~\eqref{eq:local_adaptation}, this redistribution changes the spatial deployment of adaptive dissipation and provides the basis for the synchronization pathways examined next.
	
	\subsection{Synchronization Pathway Reorganization}
	
	Figure~\ref{fig:fig42representativegainerror} compares three representative bias levels. To characterize the network response, we use
	\begin{equation}
		\begin{aligned}
			E(t)
			&=
			\left[
			\frac{1}{N}\sum_{i=1}^{N}
			\bigl(x_i(t)-\bar{x}(t)\bigr)^2
			\right]^{1/2},\\
			J_e&=\int_0^T E^2(t)\,\mathrm{d}t,\qquad
			J_u=\int_0^T\frac{1}{N}\sum_{i=1}^{N}u_i^2(t)\,\mathrm{d}t .
		\end{aligned}
		\label{eq:finite_horizon_metrics}
	\end{equation}
	Here, \(E(t)\) measures the network-wide voltage dispersion, \(J_e\) the accumulated incoherence, and \(J_u\) the network-averaged control effort over the observation window. Nodes are ranked by degree into the top \(10\%\) hubs, middle \(40\%\), and peripheral \(50\%\), denoted by \(\mathcal{Q}_{\mathrm H}\), \(\mathcal{Q}_{\mathrm M}\), and \(\mathcal{Q}_{\mathrm P}\), respectively. For \(q\in\{\mathrm H,\mathrm M,\mathrm P\}\), the degree-class mean gain and RMS voltage dispersion are
	\begin{equation}
		\begin{aligned}
			\bar d_q(t)
			&=\frac{1}{|\mathcal{Q}_q|}
			\sum_{i\in\mathcal{Q}_q}d_i(t),\\
			E_q(t)
			&=
			\left[
			\frac{1}{|\mathcal{Q}_q|}
			\sum_{i\in\mathcal{Q}_q}
			\bigl(x_i(t)-\bar x(t)\bigr)^2
			\right]^{1/2}.
		\end{aligned}
		\label{eq:degree_class_metrics}
	\end{equation}
	and the network-mean gain is \(\bar d(t)=N^{-1}\sum_{i=1}^N d_i(t)\).
	
	\begin{figure*}[!t]
		\centering
		\includegraphics[width=0.90\textwidth]
		{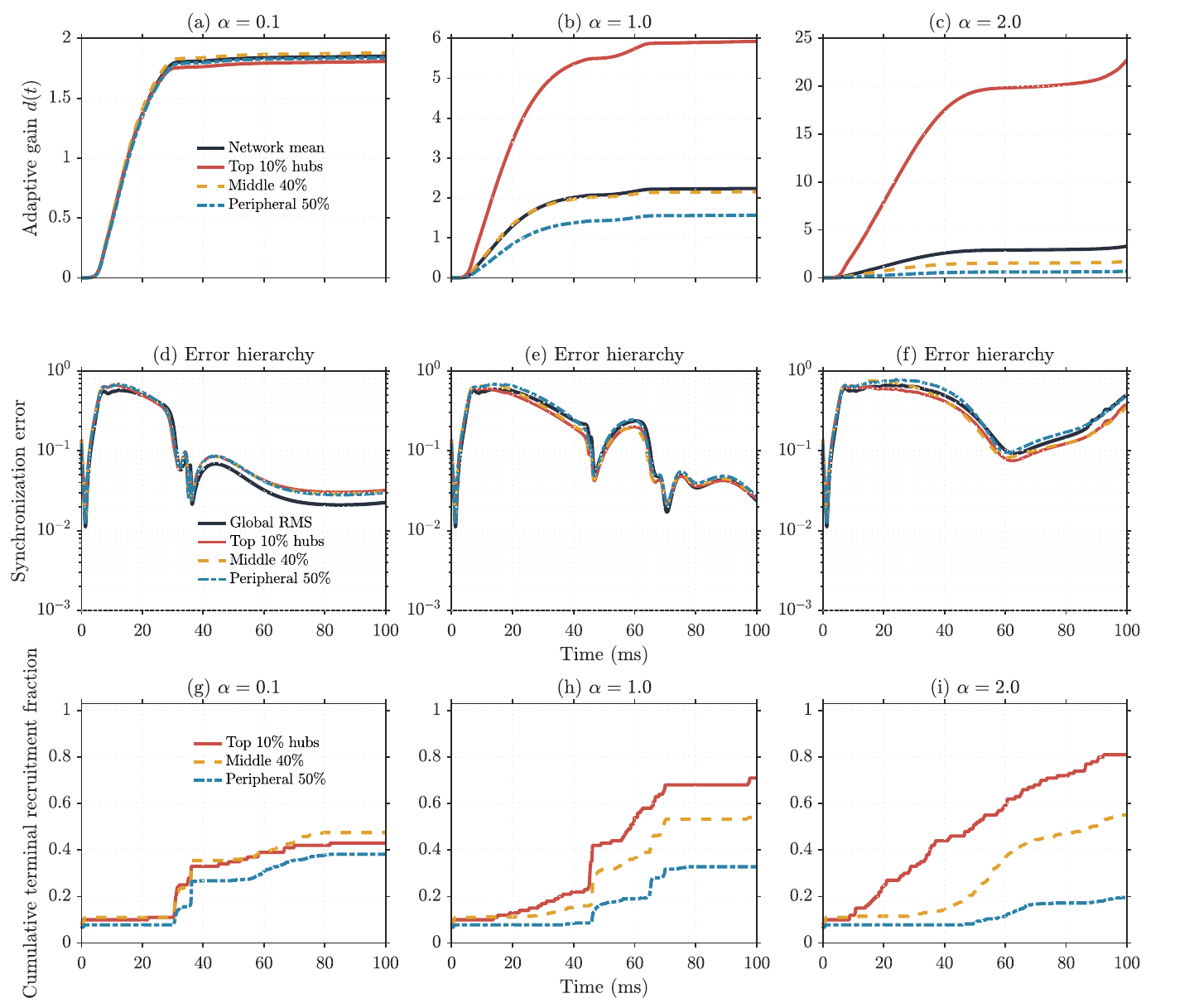}
		\caption{
			Degree-resolved synchronization pathways for \(\alpha=0.1,1,2\).
			The top row shows adaptive gains for the network mean and three degree classes, the middle row shows the corresponding error dynamics, and the bottom row shows cumulative first-sustained recruitment for the hub, middle-degree, and peripheral groups.
		}
		\label{fig:fig42representativegainerror}
	\end{figure*}
	
	At \(\alpha=0.1\), the resource field is broadly distributed and the adaptive gains of the three degree classes remain close. Their error curves contract along similar trajectories, and no class maintains a persistent lead. The transient is therefore close to a collective contraction.
	
	At \(\alpha=1\), the gain curves separate clearly. The hub population accumulates dissipation earlier, while enough resource remains outside the highest-degree group for the middle and peripheral populations to continue adapting. The error traces acquire a persistent temporal order: hubs contract first, the middle-degree group follows, and the periphery is recruited later. The contraction--rebound episodes reflect the chaotic bursting dynamics, but the degree ordering survives these excursions.
	
	At \(\alpha=2\), the hub gains become much larger than those of the other degree classes. Hub errors contract under strong local dissipation, whereas peripheral errors remain appreciable for much of the \(100\,\mathrm{ms}\) window. The resulting hierarchy is stronger, but recruitment propagates through the network more slowly.
	
	The bottom row of Fig.~\ref{fig:fig42representativegainerror} is intended to resolve the onset of node capture rather than its final retention. Let
	\(\xi_i(t)=|x_i(t)-\bar{x}(t)|\). For this representative pathway analysis, we use the fixed local tolerance \(\varepsilon_{\mathrm s}=10^{-2}\) and a residence interval \(t_{\mathrm h}=2\,\mathrm{ms}\). The first-sustained recruitment time is defined as
	\begin{equation}
		\tau_i^{\mathrm S}
		=
		\inf\left\{
		t\in[0,T-t_{\mathrm h}]:
		\xi_i(\tau)\leq\varepsilon_{\mathrm s},\;
		\forall\,\tau\in[t,t+t_{\mathrm h}]
		\right\}.
		\label{eq:first_sustained_recruitment_time}
	\end{equation}
	The corresponding cumulative fraction for each degree class is
	\begin{equation}
		R_q^{\mathrm S}(t;\alpha)
		=
		\frac{1}{|\mathcal{Q}_q|}
		\sum_{i\in\mathcal{Q}_q}
		\mathbf{1}\!\left(\tau_i^{\mathrm S}\le t\right),
		\qquad
		q\in\{\mathrm H,\mathrm M,\mathrm P\}.
		\label{eq:first_sustained_recruitment_fraction}
	\end{equation}
	Thus, a node is recorded when it first enters the coherent neighborhood and remains there for at least \(2\,\mathrm{ms}\); a later excursion does not erase that first-sustained entry time. The three curves are nearly overlapping at \(\alpha=0.1\). At \(\alpha=1\), they separate in the order \(R_{\mathrm H}^{\mathrm S}>R_{\mathrm M}^{\mathrm S}>R_{\mathrm P}^{\mathrm S}\), while at \(\alpha=2\) the separation becomes larger and the peripheral class is recruited later. These representative trajectories therefore identify how topology bias reorganizes the order in which different structural populations are first captured by the contracting motion.
	
	\begin{figure*}[!t]
		\centering
		\includegraphics[width=0.86\textwidth]
		{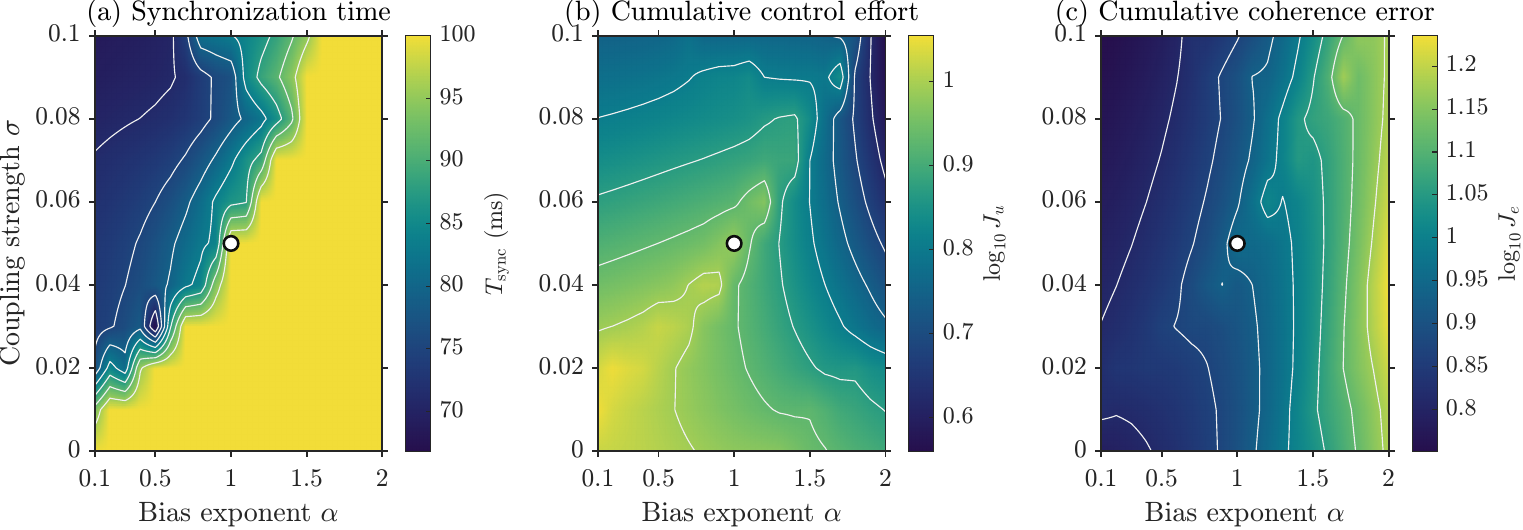}
		\caption{
			Finite-horizon response landscape in the \((\alpha,\sigma)\) plane:
			(a) synchronization time \(T_{\mathrm{sync}}\),
			(b) cumulative control effort \(\log_{10}J_u\), and
			(c) cumulative coherence error \(\log_{10}J_e\).
		}
		\label{fig:figalphasigmaresponselandscape}
	\end{figure*}
	
	The broader response is shown in Fig.~\ref{fig:figalphasigmaresponselandscape}. Here \(T_{\mathrm{sync}}\) is the earliest time at which \(E(t)\leq10^{-3}\) and remains below this level for \(2\,\mathrm{ms}\); trajectories that do not satisfy this condition within \(T=100\,\mathrm{ms}\) are right-censored at \(T\). Increasing \(\sigma\) generally shortens the transient by strengthening diffusive contraction throughout the interaction network. The response to \(\alpha\) is non-monotonic. In the strongly localized region, both \(T_{\mathrm{sync}}\) and \(J_e\) increase even though the hub weighting is stronger.
	
	The control-effort landscape gives a complementary view. Relatively small 
	\(J_u\) can coexist with long synchronization times and large accumulated 
	incoherence in the strongly localized region. The total effort alone 
	therefore does not determine the transient response; its spatial deployment 
	across the network also matters.
	
	\begin{figure*}[!t]
		\centering
		\includegraphics[width=0.96\textwidth]{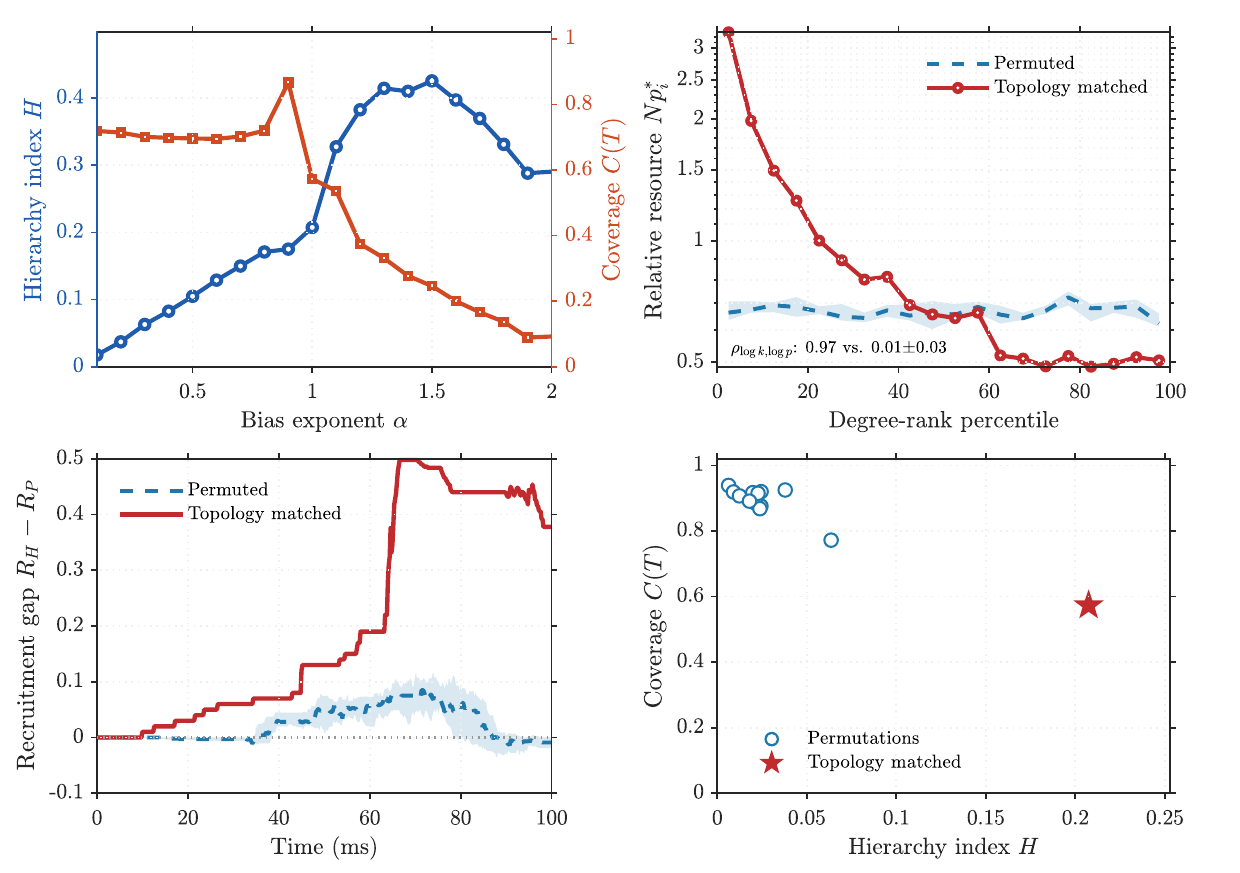}
		\caption{
			Quantification and structural test of hierarchy--coverage separation 
			under the terminal-recruitment protocol.
			(a) Hierarchy index \(H\) and finite-horizon coverage \(C(T)\) versus 
			\(\alpha\) at \(\sigma=0.05\).
			(b) Degree-ranked resource profiles for the topology-matched allocation 
			and for random permutations of the same stationary resource values.
			(c) Hub--periphery terminal-recruitment gap 
			\(R_{\mathrm H}^{\mathrm T}-R_{\mathrm P}^{\mathrm T}\) for the matched 
			allocation and the permutation ensemble.
			(d) The corresponding states in the \(H\)--\(C(T)\) plane. Twelve resource 
			permutations are used; they preserve the resource values and their 
			heterogeneity while removing their degree alignment.
		}
		\label{fig:hierarchy_coverage_alignment}
	\end{figure*}
	
	For the population-level quantification below, we use a complementary 
	terminal-retention criterion with the baseline spatial tolerance
	\begin{equation}
		\varepsilon_{\mathrm r}
		=
		2.235\times10^{-2}.
		\label{eq:terminal_recruitment_threshold}
	\end{equation}
	The sensitivity of the resulting measures to this tolerance is examined 
	explicitly in Fig.~7. A node is terminally recruited only if it enters 
	the coherent neighborhood and subsequently remains there to the end of 
	the observation window,
	\begin{equation}
		\tau_i^{\mathrm T}
		=
		\inf\left\{
		t\in[0,T]:
		\xi_i(\tau)\leq\varepsilon_{\mathrm r},\;
		\forall\,\tau\in[t,T]
		\right\}.
		\label{eq:terminal_recruitment_time}
	\end{equation}

	The corresponding degree-class recruitment fraction is
	\begin{equation}
		R_q^{\mathrm T}(t;\alpha)
		=
		\frac{1}{|\mathcal{Q}_q|}
		\sum_{i\in\mathcal{Q}_q}
		\mathbf{1}\!\left(\tau_i^{\mathrm T}\le t\right),
		\qquad
		q\in\{\mathrm H,\mathrm M,\mathrm P\}.
		\label{eq:cumulative_recruitment_fraction}
	\end{equation}
	
	The two recruitment measures answer different questions. The first-sustained time \(\tau_i^{\mathrm S}\) resolves when a node is initially captured for a nonzero residence interval, whereas the terminal time \(\tau_i^{\mathrm T}\) resolves when that node becomes permanently retained within the coherent set over the remaining finite horizon. They are therefore not expected to coincide pointwise. This distinction is particularly relevant for the contraction--rebound transients of the HR network, where a node can enter the coherent neighborhood and later leave it during a subsequent burst. Figure~\ref{fig:fig42representativegainerror} is used to resolve the ordering of initial capture, while Figs.~\ref{fig:hierarchy_coverage_alignment} and \ref{fig:hierarchy_coverage_robustness} quantify terminal retention and its sensitivity to the observation protocol.
	
	To quantify the distinction between degree ordering and network-wide terminal recruitment, we introduce
	\begin{equation}
		\begin{aligned}
			H(\alpha)
			&=
			\frac{1}{T}
			\int_0^T
			\bigl[R_{\mathrm H}^{\mathrm T}(t;\alpha)-R_{\mathrm P}^{\mathrm T}(t;\alpha)\bigr]
			\,\mathrm{d}t,\\
			C(\alpha;T)
			&=
			\frac{1}{N}
			\sum_{i=1}^{N}
			\mathbf{1}\!\left(\tau_i^{\mathrm T}\leq T\right).
		\end{aligned}
		\label{eq:hierarchy_coverage_metrics}
	\end{equation}
	The hierarchy index \(H\) measures the time-averaged terminal-recruitment lead of hubs over peripheral nodes, while \(C(\alpha;T)\) is the fraction of nodes terminally retained within the observation window.
	
	\begin{figure*}[!t]
		\centering
		\includegraphics[width=0.82\textwidth]{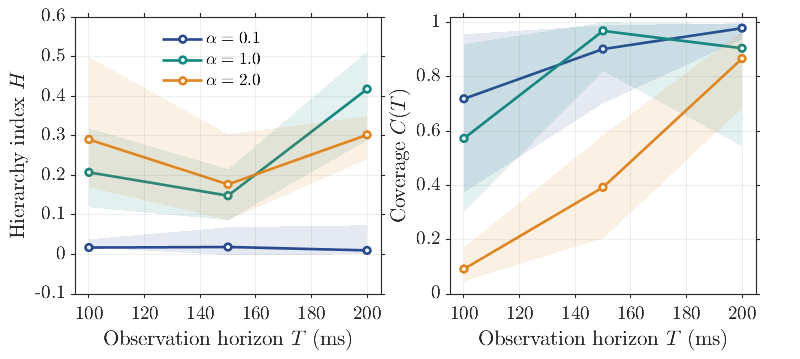}
		\caption{
			Robustness of the hierarchy--coverage distinction to the observation horizon and terminal-recruitment threshold.
			(a) Hierarchy index \(H\) and
			(b) terminal coverage \(C(T)\) for \(\alpha=0.1,1,2\) and \(T=100,150,200\,\mathrm{ms}\).
			Solid curves use the baseline terminal-recruitment threshold, and shaded bands span threshold multipliers from \(0.5\) to \(2\).
		}
		\label{fig:hierarchy_coverage_robustness}
	\end{figure*}
	
	Figure~\ref{fig:hierarchy_coverage_alignment}(a) makes the separation visible over the full \(\alpha\) scan. At \(\alpha=0.1\), \(H\simeq0.02\) and \(C(T)\simeq0.72\), consistent with weak degree ordering. The hierarchy develops as the bias increases: \(H\simeq0.20\) at \(\alpha=1\), reaches about \(0.42\) near \(\alpha=1.5\), and remains substantial at \(\alpha=2\) with \(H\simeq0.29\). Coverage follows a different trend. It reaches approximately \(0.86\) near \(\alpha=0.9\), decreases to about \(0.57\) at \(\alpha=1\), and falls to about \(0.09\) at \(\alpha=2\) within the \(100\,\mathrm{ms}\) window. The strongest degree separation and the broadest terminal recruitment therefore occur at different bias levels.
	
	The permutation experiment in Fig.~\ref{fig:hierarchy_coverage_alignment}(b)--(d) tests whether this hierarchy can be attributed to resource heterogeneity alone. Let \(a_i=\log k_i\) and \(b_i=\log p_i^*\). Their Pearson correlation coefficient is
	\begin{equation}
		\rho_{\log k,\log p}
		=
		\frac{\sum_{i=1}^{N}(a_i-\bar a)(b_i-\bar b)}
		{\left[\sum_{i=1}^{N}(a_i-\bar a)^2
			\sum_{i=1}^{N}(b_i-\bar b)^2\right]^{1/2}},
		\label{eq:resource_degree_correlation}
	\end{equation}
	where \(\bar a\) and \(\bar b\) are node averages. At \(\alpha=1\), the topology-matched field gives \(\rho_{\log k,\log p}=0.97\). Randomly reassigning the same \(p_i^*\) values to nodes leaves the total resource and the resource-value distribution unchanged, but reduces this association to \(0.01\pm0.03\). The recruitment consequence is pronounced. The matched case gives \(H\simeq0.20\), whereas all twelve permutations remain below \(H=0.07\) and cluster close to zero. At the same time, the matched terminal coverage is about \(0.57\), while the permuted cases lie approximately between \(0.77\) and \(0.94\). Thus, the hub-led hierarchy is tied to the alignment between resource placement and structural centrality rather than to resource heterogeneity by itself.

	The same terminal-recruitment ordering persists when the observation horizon and recruitment threshold are varied, as shown in Fig.~\ref{fig:hierarchy_coverage_robustness}. With the baseline threshold, the weak-bias case remains almost nonhierarchical, with \(H=0.02\), \(0.02\), and \(0.01\) for \(T=100\), \(150\), and \(200\,\mathrm{ms}\), respectively. For \(\alpha=2\), \(H\) remains clearly positive over the same horizons (\(0.29\), \(0.18\), and \(0.30\)), while terminal coverage rises from \(0.09\) to \(0.39\) and then \(0.87\). The threshold bands change the numerical values but preserve the contrast between weak-bias collective contraction and the degree-ordered transients at larger bias. The recovery of terminal coverage with a longer horizon further clarifies the strong-bias regime: excessive localization delays the final retention of peripheral nodes rather than preventing their eventual incorporation into the coherent motion.
	
	\begin{figure*}[!t]
		\centering
		\includegraphics[width=0.98\textwidth]
		{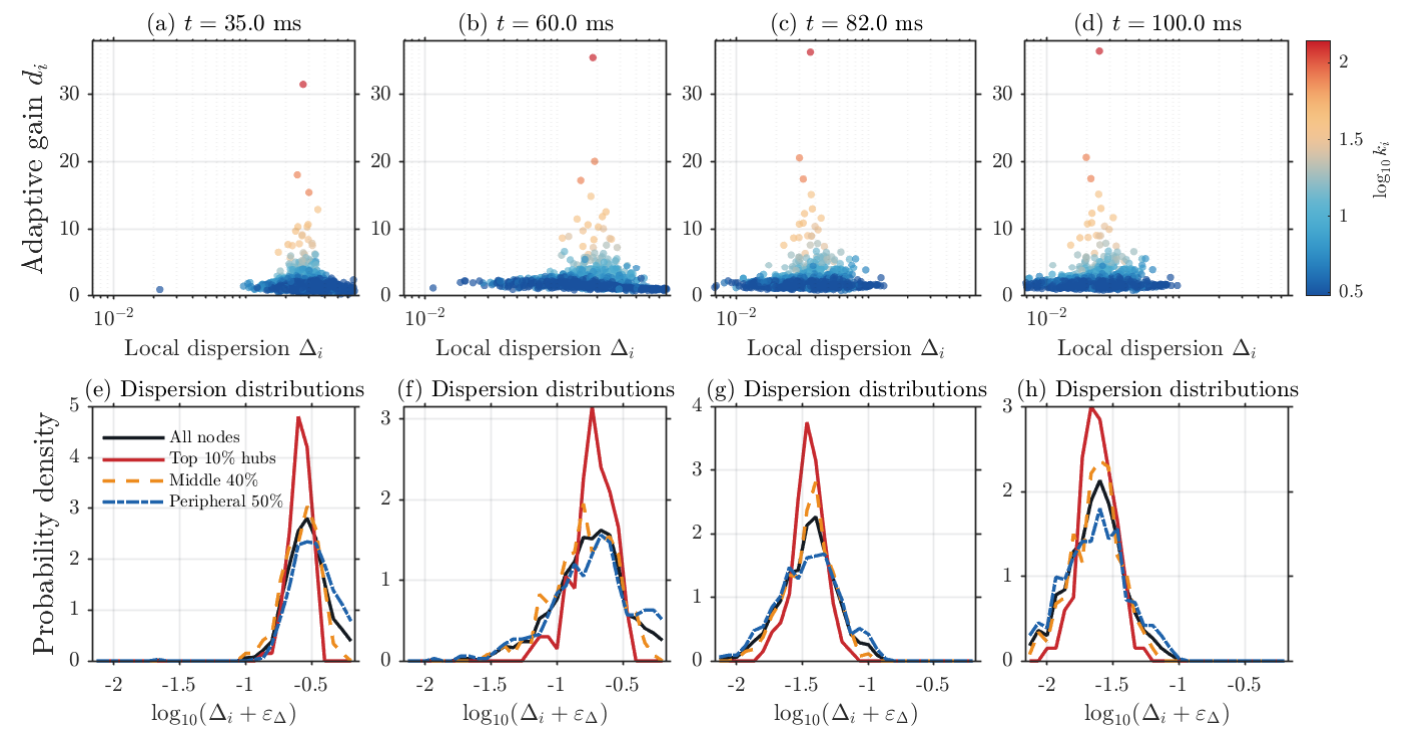}
		\caption{
			Node-wise gain--dispersion evolution for the hierarchical-recruitment regime at \(\alpha=1\).
			The upper row shows instantaneous distributions in the \((\Delta_i,d_i)\) plane, with color indicating \(\log_{10}k_i\); the lower row shows the corresponding dispersion distributions for the full network and three degree classes.
		}
		\label{fig:fig42recruitmentsnapshotsalpha1}
	\end{figure*}
	
	Figure~\ref{fig:fig42recruitmentsnapshotsalpha1} resolves the same mechanism in the gain--dispersion plane. At the earlier snapshots, most nodes still occupy a broad finite-dispersion region, while the largest adaptive gains are already concentrated predominantly among high-degree nodes. For comparable local disagreement, the adaptive law weights gain accumulation by the stationary resource \(p_i^*\), so topology-biased allocation gives structurally central nodes an earlier dissipative advantage.
	
	As the transient develops, the hub population moves toward the low-dispersion region first, followed by the middle-degree and peripheral groups. The lower-row distributions show the same ordering statistically: the hub distribution narrows earlier, whereas the peripheral population remains broadly dispersed for a longer interval. This progression is consistent with the first-sustained recruitment ordering in Fig.~\ref{fig:fig42representativegainerror} and with the terminal hierarchy quantified in Fig.~\ref{fig:hierarchy_coverage_alignment}.
	
	The gain and dispersion coordinates should not be interpreted as an instantaneous input--output relation. Since \(\dot d_i\geq0\), the adaptive gain records the accumulated history of local disagreement weighted by the available resource. A node may therefore retain a large gain after its instantaneous dispersion has already contracted, while bursting can subsequently regenerate local deviations. This memory provides a dynamical explanation for why first-sustained capture and terminal retention need not coincide, especially under strong resource localization.
	
	Overall, increasing \(\alpha\) does more than change the magnitude of adaptive damping. It reorganizes the temporal route by which coherence is established: weak bias produces nearly collective capture, intermediate bias generates a propagating hub-led hierarchy, and strong bias amplifies that hierarchy while delaying the terminal retention of the periphery. The two recruitment protocols separate the onset of local capture from its persistence and thereby expose the contraction--rebound structure that would be hidden by either measure alone.
	
	\subsection{Node-Level Dynamics of Hierarchical Recruitment}
	
	The preceding results locate the hierarchy at the level of degree classes. We next examine how it develops across individual nodes. The local dispersion
	\begin{equation}
		\Delta_i(t)
		=
		\left[
		\frac{1}{k_i}
		\sum_{j\in\mathcal{N}_i}
		\bigl(x_j(t)-x_i(t)\bigr)^2
		\right]^{1/2}
		\label{eq:local_dispersion}
	\end{equation}
	measures the instantaneous disagreement of node \(i\) with its interaction neighborhood. We focus on \(\alpha=1\), where the degree ordering is clear but recruitment still reaches substantial parts of the middle and peripheral populations.

	Figure~\ref{fig:fig42recruitmentsnapshotsalpha1} resolves the transient in the \((\Delta_i,d_i)\) plane. At \(t=35\,\mathrm{ms}\), the network still occupies a broad finite-dispersion region, but the largest gains already belong predominantly to high-degree nodes. This is consistent with the adaptive law: for comparable local disagreement, larger \(p_i^*\) produces faster gain accumulation.
	
	By \(t=60\) and \(82\,\mathrm{ms}\), high-degree nodes increasingly occupy the low-dispersion region while retaining the gains accumulated earlier in the transient. The middle-degree population contracts later, and the peripheral population remains more broadly dispersed. At \(t=100\,\mathrm{ms}\), all three distributions have shifted toward lower dispersion, while the degree-dependent gain hierarchy remains visible. The lower-row densities show the same progression: the hub distribution narrows first, followed by the middle-degree and peripheral groups.
	
	The persistence of large \(d_i\) after \(\Delta_i\) has decreased is also visible in the snapshots. Since \(\dot d_i\geq0\), gains accumulated during earlier transverse excursions are retained. Subsequent contraction--rebound episodes therefore occur under the dissipation accumulated during the preceding incoherent stage. This node-level evolution is consistent with the class-level recruitment sequence in Fig.~\ref{fig:fig42representativegainerror}.
	
	As a final check, we examine the motion after synchronization. For an infinitesimal perturbation \(\delta X(t)\) along a trajectory, the finite-time largest Lyapunov exponent is
	\begin{equation}
		\lambda_{\max}(T)
		=
		\frac{1}{T}\ln\frac{\|\delta X(T)\|}{\|\delta X(0)\|},
		\qquad
		\lambda_{\max}=\lim_{T\to\infty}\lambda_{\max}(T).
		\label{eq:largest_lyapunov_exponent}
	\end{equation}
	A positive \(\lambda_{\max}\) indicates exponential separation of nearby trajectories. As a complementary geometric measure, for \(M\) sampled states \(\{X_m\}_{m=1}^{M}\), the correlation integral is
	\begin{equation}
		C(r)=\frac{2}{M(M-1)}
		\sum_{m<n}\Theta\!\left(r-\|X_m-X_n\|\right),
		\label{eq:correlation_integral}
	\end{equation}
	where \(\Theta(\cdot)\) is the Heaviside function. Within a scaling region,
	\begin{equation}
		C(r)\sim r^{D_2},
		\qquad
		D_2=\frac{\mathrm d\log C(r)}{\mathrm d\log r},
		\label{eq:correlation_dimension}
	\end{equation}
	where \(D_2\) is the correlation dimension. Figure~\ref{fig:fractals} compares these two diagnostics for an isolated HR oscillator and a node sampled after synchronization.
	
	\begin{figure}[t]
		\centering
		\includegraphics[width=0.88\columnwidth]{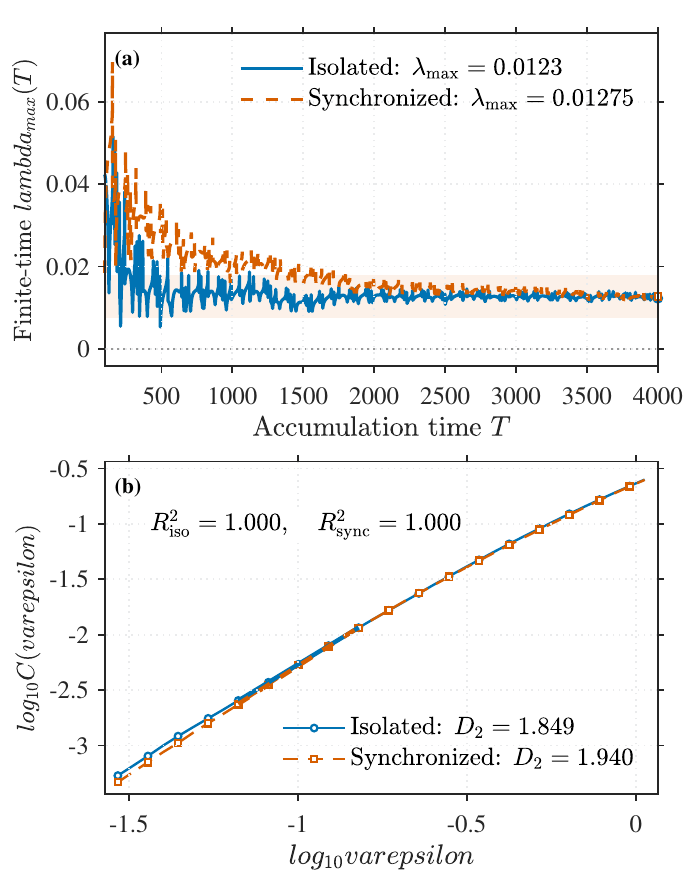}
		\caption{Chaotic dynamics before and after network synchronization.
			(a) Finite-time largest Lyapunov exponents of an isolated HR oscillator and a node sampled after synchronization.
			(b) Correlation-integral scaling of the two trajectories over the same fitting interval.}
		\label{fig:fractals}
	\end{figure}
	
	The largest Lyapunov exponent remains positive in both cases, with \(\lambda_{\max}\approx0.0123\) for the isolated oscillator and \(0.01275\) after synchronization. The corresponding correlation dimensions are \(D_2\approx1.849\) and \(1.940\). The close values indicate that synchronization suppresses transverse disagreement without replacing the intrinsic low-dimensional chaotic motion by a periodic response.
	
	This interpretation also follows from the synchronized manifold itself. When \(\eta_i=0\) and \(Lx=0\), the adaptive feedback and diffusive coupling vanish along the common trajectory. The additional dissipation therefore acts on transverse deviations, while the tangential HR dynamics remain unchanged.
	
	\FloatBarrier
	
	\section{Conclusion}
	
	This work investigated synchronization of chaotic Hindmarsh--Rose oscillators under finite topology-dependent resource allocation. A degree-biased transport process was coupled with local adaptive feedback, and the convergence of the resource field and the synchronization manifold were established analytically.
	
	The numerical results distinguish hierarchy strength from recruitment effectiveness. Weak bias produces nearly collective contraction, intermediate bias gives rise to hub-initiated hierarchical recruitment, and stronger localization further increases the degree hierarchy while slowing recruitment toward peripheral nodes. The associated non-monotonic response is consistent with a tradeoff between hub localization and network-wide dissipation coverage.
	
	The synchronized collective state remains chaotic, indicating that the adaptive feedback suppresses transverse instability without destroying the intrinsic dynamics on the synchronization manifold.

	\section*{ACKNOWLEDGMENTS}
	
	This work was supported by the Natural Science Foundation of Inner Mongolia
	Autonomous Region (Grant Nos. 2025LHMS06020, 2024MS07012 and 2025MS01018) and the
	High-Quality Development Special Research Fund of Inner Mongolia University
	of Finance and Economics (Grant Nos. NCXKY25046 and NCXKY25095).
	
	\section*{AUTHOR DECLARATIONS}
	
	\subsection*{Conflict of Interest}
	
	The authors declare no conflict of interest.
	
	\subsection*{Ethics Approval}
	
	Ethics approval is not required.
	
	\subsection*{Author Contributions}
	
	All authors have made equal contributions.
	
	\section*{DATA AVAILABILITY}
	
	Authors have no data available.
	
	\section*{GENERATIVE AI STATEMENT}
	
	The authors used ChatGPT 
	for language refinement and generating the schematic graph(see Fig.~1). No other figures were generated using
	artificial intelligence. 
	
	% TODO for the Chaos revision stage: add the required data availability
	% statement once the data-sharing route has been confirmed.
	
	% REVTeX selects the numerical AIP bibliography style. Put mybibfile.bib in
	% the same folder as this .tex file.
	\bibliography{mybibfile}
	
\end{document}